\documentclass[onefignum,onetabnum]{siamonline171218}

\usepackage{lipsum}
\usepackage{amsfonts}
\usepackage{graphicx}
\usepackage{epstopdf}
\usepackage{algorithmic}
\usepackage{xcolor}
\usepackage[caption=false]{subfig} 
\ifpdf
  \DeclareGraphicsExtensions{.eps,.pdf,.png,.jpg}
\else
  \DeclareGraphicsExtensions{.eps}
\fi

\usepackage{enumitem}
\setlist[enumerate]{leftmargin=.5in}
\setlist[itemize]{leftmargin=.5in}

\newsiamremark{remark}{Remark}
\newsiamremark{hypothesis}{Hypothesis}
\crefname{hypothesis}{Hypothesis}{Hypotheses}
\newsiamthm{Corollary}{Corollary}
\newsiamthm{Definition}{Definition}
\newsiamthm{Assumption}{Assumption}
\headers{Time-inconsistent mean field and n-agent games}{Zongxia Liang and Keyu Zhang}

\title{Time-inconsistent mean field and n-agent games under relative performance criteria\thanks{Submitted to the editors DATE.
\funding{This work was funded by the National Natural Science Foundation of China, grants 12271290 and 11871036.}}}

\author{Zongxia Liang\thanks{Department of Mathematical Sciences, Tsinghua University, Beijing, 100084 People’s Republic of China
  (\email{liangzongxia@mail.tsinghua.edu.cn}). }
\and Keyu Zhang\thanks{Department of Mathematical Sciences, Tsinghua University, Beijing, 100084 People's Republic of China
  (\email{zhangky21@mails.tsinghua.edu.cn}).}
}

\usepackage{amsopn}

\makeatletter
\newcommand*{\addFileDependency}[1]{
  \typeout{(#1)}
  \@addtofilelist{#1}
  \IfFileExists{#1}{}{\typeout{No file #1.}}
}
\makeatother

\ifpdf
\hypersetup{
  pdftitle={Time-inconsistent mean field and $n$-agent games under relative performance criteria},
  pdfauthor={Zongxia Liang and Keyu Zhang}
}
\fi

\begin{document}

\maketitle

\begin{abstract}
  In this paper we study a time-inconsistent portfolio optimization problem for competitive agents with CARA utilities and non-exponential discounting. The utility of each agent depends on her own wealth and consumption as well as the relative wealth and consumption to her competitors. Due to the presence of a non-exponential discount factor, each agent's optimal strategy becomes time-inconsistent. In order to resolve time-inconsistency, each agent makes a decision in a sophisticated way, choosing  open-loop equilibrium strategy in response to the strategies of all the other agents. We construct explicit solutions for the $n$-agent games and the corresponding mean field games (MFGs) where the limit of former yields the latter. This solution is unique in a special class of equilibria. 
\end{abstract}



\section{Introduction}
Due to the fact that peer interaction sometimes makes remarkable impacts on agent's decision making, portfolio games, as a game-theoretic extension of classical Merton problem \cite{1970Optimum}, have received considerable attention in recent years.  Relative performance is becoming an appealing way to model the interaction because of the tractability in mathematics and excellent economic motivations. The literature on portfolio games with relative performance concerns dates by to Espinosa and Touzi \cite{2013OPTIMAL}, where they consider $n$-agent games with portfolio constrains under the CARA utility by investigating the associated quadratic BSDEs systems. Lacker and Zariphopoulo \cite{D2017Mean} consider the portfolio games for  asset specialized agents in log-normal markets under both CARA and CRRA relative performance criteria. The constant Nash equilibrium and mean field equilibrium (MFE) are explicitly constructed.  In a similar fashion, Lacker and Soret \cite{2020Many} extend the problem by incorporating the dynamic consumption. Bo et al. \cite{Bo2021MeanFG} revisit the MFGs and the $n$-agent games under CRRA relative performance by allowing risky assets to have contagious jumps. A deterministic MFE in an analytical form is obtained by using the FBSDE and stochastic maximum principle. Furthermore, an approximate Nash equilibrium  for the $n$-agent games is constructed. Recently, Fu and Zhou \cite{2022meanfieldpg} study the mean field portfolio games with general market parameters. A one-to-one correspondence between the Nash equilibrium and the solution to some FBSDE is established by martingale optimality principle.  

Another research direction with fruitful outcomes is time-inconsistent control problem, where the  Bellman  optimality principle  does not hold.
There are many important problems in mathematical finance and economics incurring time-inconsistency, for example,  the mean-variance selection problem and the investment-consumption problem with non-exponential discounting.  The main approaches to handle time-inconsistency are to search for, instead of optimal strategies, time-consistent equilibrium strategies within a game-theoretic framework. Ekeland and Lazrak \cite{2008Equilibrium} and Ekeland and Pirvu \cite{Ivar2008Investment} introduce the precise definition of the equilibrium strategy in continuous-time setting for the first time. Björk et al. \cite{2017On} derive an extended HJB equation to determine the equilibrium strategy in a Markovian setting. Yong \cite{2012Time} introduces the so-called equilibrium HJB equation to construct the equilibrium strategy  in a multi-person differential game framework with a hierarchical structure. The solution concepts considered in \cite{2017On, 2012Time} are closed-loop equilibrium strategies and the methods to handle time-inconsistency are extensions of the classical dynamic programming approaches. In contrast to the aforementioned literature,  Hu et al. \cite{Hu-2012} introduce the concept of open-loop equilibrium control by using a spike variation formulation, which is different from the closed-loop equilibrium concepts. The open-loop equilibrium control is characterized by a flow of FBSDEs, which is deduced by a duality method in the spirit of Peng's stochastic maximum principle. Some recent studies devoted to the open-loop equilibrium concept can be found in \cite{2020openloop, EP14900367320210201, 2020Uniqueness, 2021Time}. Specially, Alia et al. \cite{EP14900367320210201}, closely related to our paper, study a time-inconsistent investment-consumption problem under a general discount function, and obtain an explicit representation of the equilibrium strategies for some special utility functions, which is different from most of existing literature on the time-inconsistent investment-consumption problem, where the feedback equilibrium strategies are derived via several complicated nonlocal ODEs; see, e.g., \cite{MS-2010, Bjork}.

To the best of our knowledge, the $n$-agent games and MFGs under relative performance when the non-exponential discounting is considered have not been studied before. The constant discount rate is a common assumption in classical portfolio management problems under discounted utility which suggests the discount function should be exponential. But results from experimental studies contradict this assumption, indicating that agents may be impatient in the face of  choices in the short term but be patient when choosing between long-term alternatives; see, e.g., \cite{Ainslie-1975}. Therefore, it is interesting to investigate the non-exponential discounting case.

Our paper aims to contribute to the literature of the aforementioned $n$-agent games and MFGs under relative performance by considering  general discount functions. Specially, 
 the discount  function only needs to satisfy some weak conditions; see \cref{disc} in Section \ref{sect2}. Meanwhile, we adopt the asset specialization framework in \cite{D2017Mean, Bo2021MeanFG} with a common noise. A lot of literature on equilibrium strategy under non-exponential discounting suggest that the investment strategy is independent of discount function,  which indicates that time-inconsistency does not influence agent's portfolio policy; see, e.g., \cite{MS-2010, Bjork}. Thus instead of portfolio games,  we incorporate consumption in the same spirit of \cite{2020Many} and focus on CARA relative performance utility. As opposed to the previous works, the presence of non-exponential discounting gives rise to time-inconsistency, which induces the failure of the principle of optimality. In order to resolve time-inconsistency, we replace each agent's optimal strategy in time-consistent setting by its open-loop equilibrium (consistent) strategy in time-inconsistent setting. As a result, there are two levels of game-theoretic reasoning intertwined.  (1) The intra-personal equilibrium among the agent's current and future selves. (2) The equilibrium among $n$-agents. For tractability, we search only for \textit{DF equilibrium strategy}, see \cref{DF}.  To construct the \textit{DF equilibrium strategy}, we first characterize the  open-loop consistent control for each single agent given an arbitrary (but fixed) choice of competitors' controls by a FBSDE system.  By assuming that all the other agents choose \textit{DF strategies}, see  \cref{DFs}, we derive a closed form representation of the strategy via a PDE system. We then construct the desired equilibrium by solving a fixed point problem.   While for MFG, the solution technique is analogous to the $n$-agent setting and the resulting MFE takes similar forms as its $n$-agent counterparts.

The contributions of our paper are as follows: first, as far as we know, our work is the first paper to incorporate consumption into CARA portfolio game with relative performance concerns.  Portfolio games under relative consumption are generally underexplored, with the exception of \cite{2020Many}, which focuses solely on  CRRA utilities under zero discount rate. Our paper fills this gap in the literature. Second, our work can be viewed as a game-theoretic extension of the exponential utility case presented in \cite{EP14900367320210201}. In the special case of a single stock, the DF equilibrium strategy takes the same form as the open-loop equilibrium in \cite{EP14900367320210201}, but with a modified risk tolerance.  This \textit{effective risk tolerance parameter} has already appeared in some works on a similar topic but with a time-consistent model; see, e.g., \cite{D2017Mean, 
 2021Meanfieldito}. Moreover, in the case with constant discount rate, the equilibrium reduces to the solution of classical Merton problem, which indicates that the equilibrium concept in our paper is the natural extension of equilibrium in classical time-consistent setting to time-inconsistent setting.  Third, our work also provides a new explicitly solvable mean field game model. Since the pioneering works by \cite{Lasry-2007, 2006Large}, MFGs have been actively studied and widely applied in economics, finance and engineering. To name a few recent developments in theories and application, we refer to \cite{aCarmona-2013,bCarmona-2013,Carmona-2015, book:2191018,book:2203663} among others.  However, few studies combine the MFGs with time-inconsistency problem, except some linear quadratic examples; see, e.g., \cite{Ni-2017,Bensoussan-2013}.  Our result adds a new explicitly solvable non-LQ example to intersection of these two fields.
 
 The rest of paper is organized as follows. In Section \ref{sect2},  we formulate and solve the $n$-agent games under CARA relative preference and a general discount function. Then, in Section \ref{sect3}, we study the infinite population counterpart of this problem, and identify resulting MFE agree with the limiting expressions from the $n$-agent games. Section \ref{discussion} presents some qualitative comments on the MFE. Some conclusion remarks and future research directions are given in Section \ref{sect5}. Finally, the proofs of some auxiliary results are  in Appendix.
 
\textit{Notations}. For any Euclidean space $E$ with norm $|\cdot|$, and any $t\in[0,T]$, denote
  
  $L_{\mathcal{F}_{t}}^{2}(\Omega;E)$: the set of $E$-valued $\mathcal{F}_{t}$-measurable random variables X, such that \[\mathbb{E}\left[\left|X\right|^{2}\right]<\infty.\]

  $\mathbb{S}^{2}_{\mathbb{F}}(t,T;E)$: the space of $E$-valued $\mathbb{F}$-adapted and continuous processes Y with \[\mathbb{E}\left[\sup_{s\in[t,T]}\left|Y_{s}\right|^{2}\right]<\infty.\]
 
 	$\mathbb{H}^{2}_{\mathbb{F}}(t,T;E)$:  the space of $E$-valued $\mathbb{F}$-progressively measurable processes Z with \[\mathbb{E}\left[\int_{t}^{T}\left|Z_{s}\right|^{2}ds\right]<\infty.\]

\section{The $n$-agent games}\label{sect2}

In this section, we consider the $n$-agent games. The market model is same as \cite{D2017Mean} and each agent invests in their own specific stock or in a common riskless
bond which offers zero interest rate. The common time horizon for all agents is $[0,T]$ with $T>0$.  The price process of stock $i$, in which only agent $i$ trades, follows the following stochastic differential equation (SDE)
\begin{equation}
	\frac{dS_{t}^{i}}{S_{t}^{i}}=\mu_{i}dt+\nu_{i}dW_{t}^{i}+\sigma _{i}dB_{t}\label{S}
\end{equation}
with constant parameters $\mu_{i}>0$, $\sigma_{i}\geq0$, and $\nu_{i}\geq0$ with $\sigma_{i}+\nu_{i}>0$, 
where the Brownian motions $W^{1},\cdots,W^{n}$ and $B$ are independent  on a filtered probability space $(\Omega,\mathcal{F},\mathbb{F},\mathbb{P})$, in which the natural filtration $\mathbb{F}:=(\mathcal{F}_{t})_{t\in[0,T]}$ is generated by these $n+1$ Brownian motions.

We recall the  single stock case, corresponding to the situation where $\mu_{i}=\mu$, $\nu_{i}=0$, and $\sigma_{i}=\sigma$, for all $i=1,\cdots,n$ and for some $\mu$, $\sigma>0$ independent of $i$.

Each agent $i$ {whose initial wealth at time $t_{0}\in[0,T)$ is $x^{i}_{0}\in\mathbb{R}$} trades according to a self-financing strategy, $\pi^{i}=\{\pi_{t}^{i}, t_{0}\leq t\leq T \}$,  representing the  amount invested in the stock $i$, and consumes at a consumption rate $c^{i}=\{c^{i}_{t}, t_{0}\leq t\leq T \}$. Then the $i$-th agent's wealth process $X^{i}=\{X^{i}_{t}, t_{0}\leq t\leq T\}$  is 
\begin{equation}\label{wealth}
	dX^{i}_{t}=\pi_{t}^{i}\left(\mu_{i}dt+\nu_{i}dW_{t}^{i}+\sigma _{i}dB_{t}\right)-c^{i}_{t}dt,\quad X^{i}_{{t_{0}}}=x^{i}_{0}\in\mathbb{R}.
\end{equation}
 Now, we introduce the admissible control as follows:
\begin{Definition}[Admissible control]

A control $(\pi,c)$ is said to be admissible over $[t,T]$ if $(\pi,c)\in\mathbb{H}^{2}_{\mathbb{F}}(t,T;\mathbb{R})\times\mathbb{H}^{2}_{\mathbb{F}}(t,T;\mathbb{R})$. For brevity, we denote $\mathcal{A}_{t}$ as the set of all admissible controls over $[t,T]$.
\end{Definition}

We focus on the DF (deterministic-feedback) strategy defined as follows:
\begin{Definition}[DF strategy]\label{DFs}
	
An $n$-tuple of pairs $\left(\Pi^{i}, C^{i}\right)^{n}_{i=1}$ is said to be a DF strategy, if for each $i\in\left\{1,\cdots,n\right\}$, $\Pi^{i}:[0,T]\rightarrow\mathbb{R}$ is a continuous function and $C^{i}:[0,T]\times\mathbb{R}^{n}\rightarrow\mathbb{R}$ is of an affine form:
	$$C^{i}(t,x)=\sum_{k=1}^{n}p^{i,k}(t)x^{k}+q^{i}(t),\,\ (t,x)\in[0,T]\times\mathbb{R}^{n},$$
	for some continuous functions $p^{i,k},\ q^{i}:[0,T]\rightarrow\mathbb{R}$. We denote the set of DF strategies by $\mathcal{S}^{n}$. Moreover, a DF strategy $\left(\Pi^{i}, C^{i}\right)^{n}_{i=1}\in\mathcal{S}^{n}$ is said to be simple if for every $i\in\left\{1,\cdots,n\right\}$, the function $[0,T]\times\mathbb{R}^{n}\ni\left(t,x=\left(x^{1},\cdots,x^{n}\right)\right)\rightarrow C^{i}(t,x)\in\mathbb{R}$ does not depend on $\left(x^{k}\right)_{k\neq i}$.
\end{Definition}
Let a DF strategy $\left(\Pi^{i}, C^{i}\right)^{n}_{i=1}\in\mathcal{S}^{n}$  be given. Then, for each initial condition $\left(t_{0},x_{0}\right)\in[0,T)\times\mathbb{R}^{n}$ with $x_{0}=(x^{1}_{0},\cdots,x^{n}_{0})$, we can solve the closed-loop system for $X=\left(X^{1},\cdots,X^{n}\right)$:
\begin{equation}\label{NX}
		\begin{cases}
			dX^{i}_{t}=\left(\Pi^{i}(t)\mu_{i}-C^{i}(t,X_{t})\right)dt+\Pi^{i}(t)\nu_{i}dW^{i}_{t}+\Pi^{i}(t)\sigma_{i}dB_{t},\, t\in[t_{0},T],\\ i=1,\cdots,n, \\
			X^{i}_{t_{0}}=x_{0}^{i},\, i=1,\cdots,n.
		\end{cases}
	\end{equation}
Then, it is easy to see that the \textsl{outcome} $\left(\pi^{i},c^{i}\right)_{i=1}^{n}$ defined by 
\begin{equation}\label{outcome}
	\pi^{i}_{t}:=\Pi^{i}(t),\, c^{i}_{t}:=C^{i}(t,X_{t}),\, t\in[t_{0},T],\, i=1,\dots,n,
\end{equation}
is in $\mathcal{A}^{n}_{t_{0}}$, that is, an $n$-tuple of admissible controls on $[t_{0},T]$.
\begin{remark}
It is worth noting that the DF strategy $\left(\Pi^{i}, C^{i}\right)_{i=1}^{n}$ does not depend on the initial condition $\left(t_{0},x_{0}\right)$, while the outcome $\left(\pi^{i},c^{i}\right)_{i=1}^{n}$ depends on $\left(t_{0},x_{0}\right)$. 

	Let $\left(\Pi^{i}, C^{i}\right)_{i=1}^{n}, \left(\tilde{\Pi}^{i}, \tilde{C}^{i}\right)_{i=1}^{n}\in\mathcal{S}^{n}$ be given.  For each initial condition $\left(t_{0},x_{0}\right)\in[0,T)\times\mathbb{R}^{n}$, denote by $\left(\pi^{i}(t_{0},x_{0}),c^{i}(t_{0},x_{0})\right)_{i=1}^{n}$ and by $\left(\tilde{\pi}^{i}(t_{0},x_{0}),\tilde{c}^{i}(t_{0},x_{0})\right)_{i=1}^{n}$ the corresponding outcomes. Fix $i\in\left\{1,\cdots,n\right\}$, and assume that $\left(\pi^{i}_{t}(t_{0},x_{0}),c^{i}_{t}(t_{0},x_{0})\right)=\left(\tilde{\pi}^{i}_{t}(t_{0},x_{0}),\tilde{c}^{i}_{t}(t_{0},x_{0})\right)$ a.s. for a.e. $t\in[t_{0},T]$ for every $\left(t_{0},x_{0}\right)\in[0,T)\times\mathbb{R}^{n}$. Then it holds that $\Pi^{i}(t)=\tilde{\Pi}^{i}(t)$ and $C^{i}(t,x)=\tilde{C}^{i}(t,x)$ for any $\left(t,x\right)\in[0,T]\times\mathbb{R}^{n}$.
\end{remark}
As we mention before, growing evidence suggests that the discount rate may not be constant, and in our work, we discuss the general discounting preferences.

\begin{Definition}[Discount Function]\label{disc}

A discount function $\lambda:[0, T]\rightarrow\mathbb{R}$ is a continuous and strictly positive function satisfying
$\lambda(0)=1$. 
\end{Definition}
 
 \begin{remark}\label{remark1}
  Definition \ref{disc} is general enough to cover some special discount functions, such as exponential discount functions (see, e.g., \cite{1970Optimum}), quasi-exponential discount functions (see, e.g., \cite{Ivar2008Investment}) and hyperbolic discount functions (see, e.g., \cite{Zhao-2014}). In fact, our result can be extended to a more general form of the discount factor as \cite{2021Time}, where the discount function $\lambda(\cdot,\cdot)$ is a positive bivariate continuous function on $\{(t,s)|0\leq t\leq s\leq T\}$ satisfying $\lambda(t,t)=1$, as the very thing we need is continuity of the discount function.
 \end{remark}

The utility function of the agent $i$ is defined as follows:
\[U_{i}\left(x^{i},\overline{x}\right):=-\exp\left\{-\frac{1}{\delta_{i}}\left(x^{i}-\theta_{i}\overline{x}\right)\right\},\, x=\left(x^{1},\dots,x^{n}\right)\in\mathbb{R}^{n},\]
where $\overline{x}:=\frac{1}{n}\sum_{i=1}^{n}x^{i}$ and the constants $\delta_{i}>0$ and $\theta_{i}\in[0,1]$ represent the personal risk tolerance and competition weight parameters, respectively. Note that the utility function also can be seen as a function of $x^{i}$ and $\overline{x}^{(i)}:=\frac{1}{n}\sum_{j\ne i}x^{j}$:
\[U_{i}\left(x^{i},\overline{x}^{(i)}\right):=-\exp\left\{-\frac{1}{\delta_{i}}\left(1-\frac{\theta_{i}}{n}\right)x^{i}+\frac{\theta_{i}}{\delta_{i}}\overline{x}^{(i)}\right\}.\]

Each agent derives a reward from their discounted inter-temporal consumption and final wealth, to be specific, for agent $i$, the expected payoff is 
\begin{equation}\label{eu}
\mathbb{E}\left[\int_{{t_{0}}}^{T}{\lambda(t-t_{0})}U_{i}\left(c^{i}_{t},\overline{c}^{(i)}_{t}\right)dt+{\lambda(T-t_{0})}U_{i}\left(X_{T}^{i},\overline{X}^{(i)}_{T}\right)\right],
\end{equation}
defined for any admissible controls {$(\pi^{i},c^{i})_{i=1}^{n}\in\mathcal{A}_{t_{0}}^{n}$}, where   $\overline{X}^{(i)}_{t}:=\frac{1}{n}\sum_{k\neq i}X^{k}_{t}$ and $\overline{c}^{(i)}_{t}:=\frac{1}{n}\sum_{k\neq i}c^{k}_{t}$.

For  $(\pi^{i},c^{i})_{i=1}^{n}\in\mathcal{A}_{t_{0}}^{n}$, we introduce the following notations:
\begin{equation*}
	\begin{split}
	J_{i}(t,X^{i}_{t},\overline{X}^{(i)}_{t},(\pi^{i},c^{i})_{i=1}^{n}):&=\mathbb{E}_{t}\left[\int_{t}^{T}\lambda(s-t)U_{i}\left(c^{i}_{s},\overline{c}^{(i)}_{s}\right)ds+\lambda(T-t)U_{i}\left(X_{T}^{i},\overline{X}^{(i)}_{T}\right)\right],
	\end{split}
\end{equation*}
 where $X^{i}$ is defined by \cref{wealth}, $\overline{X}^{(i)}$ satisfies the following SDE:
 \begin{equation}\label{overlinex}
 	\begin{cases}
 		d\overline{X}^{(i)}_{t}=(\overline{\mu\pi}^{(i)}_{t}-\overline{c}^{(i)}_{t})dt+\overline{\sigma\pi}^{(i)}_{t}dB_{t}+\frac{1}{n}\sum_{k\ne i}\nu_{k}\pi^{k}_{t}dW^{k}_{t},\, t\in[{t_{0}},T],\\ \overline{X}^{(i)}_{{t_{0}}}=\overline{x}^{(i)}_{0}:=\frac{1}{n}\sum_{k\neq i}x_{0}^{k},
 	\end{cases}
 \end{equation}
with $\overline{\mu\pi}^{(i)}_{t}:=\frac{1}{n}\sum_{k\ne i}\mu_{k}\pi^{k}_{t}$ and $\overline{\sigma\pi}^{(i)}_{t}:=\frac{1}{n}\sum_{k\neq i}\sigma_{k}\pi^{k}_{t}$, and $\mathbb{E}_{t}\left[\cdot\right]:=\mathbb{E}\left[\cdot|\mathcal{F}_{t}\right]$.

It is well known that the optimal strategy of agent $i$ turns out to be time-inconsistent as soon as discounting is non-exponential.  Hence, we assume that all agents are sophisticated, which means that they aim to find the best current action in response to their future selves' behavior. When every future self also reasons in this way, the resulting strategy is an equilibrium form which no future self has any incentive to deviate.

As in \cite{EP14900367320210201, Hu-2012}, we consider open-loop Nash equilibrium by local spike variation. For {$t\in [t_{0},T)$}, any $\mathbb{R}^{2}$-valued, $\mathcal{F}_{t}$-measurable and bounded random variable $v=(v_{1},v_{2})$, and any $\epsilon>0$, given an admissible control {$(\hat{\pi},\hat{c})\in\mathcal{A}_{t_{0}}$}, define 
{\begin{equation}	\label{eps}
	(\pi_{s}^{t,\epsilon},c_{s}^{t,\epsilon})=(\hat{\pi}_{s},\hat{c}_{s})+v\mathbf{1}_{[t,t+\epsilon)}(s)
\end{equation}
for $s\in[t_{0},T]$, where $\mathbf{1}_{[t,t+\epsilon)}$ denotes the indicator function for the interval $[t,t+\epsilon)$.}
\begin{Definition}[Open-loop consistent control]\label{OPCS}
	
Let $i\in\left\{1,\cdots,n\right\}$, { $(t_{0},x_{0})\in[0,T)\times\mathbb{R}^{n}$} and $\left(\pi,c\right)^{(i)}=\left(\pi^{k},c^{k}\right)_{k\neq i}\in{\mathcal{A}^{n-1}_{t_{0}}}$ be given. An admissible control $(\hat{\pi}^{i},\hat{c}^{i})\in{\mathcal{A}_{t_{0}}}$ is said to be an open-loop consistent control for agent $i$ with respect to the {initial condition $(t_{0},x_{0})$} in response to $\left(\pi^{k},c^{k}\right)_{k\neq i}$ if for every {$t\in[t_{0},T)$} and every sequence $\left\{\epsilon_{n}\right\}_{n\in\mathbb{N}}\subset(0,T-t)$ such that $\lim\limits_{n\rightarrow\infty}\epsilon_{n}=0$, the following local optimality condition holds:
\begin{align*}
		\limsup_{n\rightarrow\infty}\frac{1}{\epsilon_{n}}\left\{J_{i}(t,\hat{X}^{i}_{t},\overline{X}^{(i)}_{t},(\pi^{i,t,\epsilon_{n}},c^{i,t,\epsilon_{n}}),\left(\pi,c\right)^{(i)})-J_{i}(t,\hat{X}^{i}_{t},\overline{X}^{(i)}_{t},(\hat{\pi}^{i},\hat{c}^{i}),\left(\pi,c\right)^{(i)})\right\}\leq0,\,\ a.s.,
	\end{align*}
where $\hat{X}^{i}$ and $\overline{X}^{(i)}$ are defined by \cref{wealth} and \cref{overlinex}, respectively, with the {initial condition $(t_{0},x_{0})$} and the admissible controls $(\hat{\pi}^{i},\hat{c}^{i}),\left(\pi,c\right)^{(i)}$, that is, 
\begin{equation*}
	\begin{cases}
		d\hat{X}^{i}_{t}=\left(\hat{\pi}^{i}_{t}\mu_{i}-\hat{c}^{i}_{t}\right)dt+\hat{\pi}^{i}_{t}\nu_{i}dW_{t}^{i}+\hat{\pi}^{i}_{t}\sigma_{i}dB_{t},\,t\in[{t_{0}},T],\\
		\hat{X}^{i}_{{t_{0}}}=x^{i}_{0},
	\end{cases}
\end{equation*}
and
\begin{equation*}
	\begin{cases}
		d\overline{X}^{(i)}_{t}=(\overline{\mu\pi}^{(i)}_{t}-\overline{c}^{(i)}_{t})dt+\overline{\sigma\pi}^{(i)}_{t}dB_{t}+\frac{1}{n}\sum_{k\ne i}\nu_{k}\pi^{k}_{t}dW^{k}_{t}, \, t\in[{t_{0}},T],\\ \overline{X}^{(i)}_{{t_{0}}}=\overline{x}^{(i)}_{0}:=\frac{1}{n}\sum_{k\neq i}x_{0}^{k}.
	\end{cases}
\end{equation*}
\end{Definition}
\begin{remark}
	We note that there is a slight difference between our definition and the classical definition of open-loop equilibrium given by \cite{Hu-2012}. If we adapt the classical definition to our problem, the open-loop consistent control should satisfy
	\begin{align*}
		\limsup_{\epsilon\downarrow0}\frac{1}{\epsilon}\left\{J_{i}(t,\hat{X}^{i}_{t},\overline{X}^{(i)}_{t},(\pi^{i,t,\epsilon},c^{i,t,\epsilon}),\left(\pi,c\right)^{(i)})-J_{i}(t,\hat{X}^{i}_{t},\overline{X}^{(i)}_{t},(\hat{\pi}^{i},\hat{c}^{i}),\left(\pi,c\right)^{(i)})\right\}\leq0,\,\ a.s.,
	\end{align*}
where $\{J_{i}(t,\hat{X}^{i}_{t},\overline{X}^{(i)}_{t},(\pi^{i,t,\epsilon},c^{i,t,\epsilon}),\left(\pi,c\right)^{(i)})\}_{\epsilon>0}$ is an uncountable family of random variables and the a.s. limit as $\epsilon\downarrow0$ may not be well-defined. Hence, we avoid to use this definition in our problem.
\end{remark}

We introduce the following technical integrability condition to ensure the uniqueness of the open-loop consistent control.
{\begin{Definition}[Integrability condition]\label{Assumption2.3}
	
Let $(t_{0},x_{0})\in[0,T)\times\mathbb{R}^{n}$ be given. Admissible controls $(\pi^{i},c^{i})_{i=1}^{n}\in\mathcal{A}_{t_{0}}^{n}$ satisfy the integrability condition over $[t_{0},T]$ if for every $i\in\left\{1,\cdots,n\right\}$,
	\begin{equation}\label{interequ}
		\mathbb{E}\left[\int_{t_{0}}^{T}|U_{i}(c^{i}_{t},\overline{c}^{(i)}_{t})|^{2}dt+|U_{i}(X_{T}^{i},\overline{X}^{(i)}_{T})|^{2}\right]<\infty,
	\end{equation}\
where $X^{i}$ and $\overline{X}^{(i)}$ are defined by \cref{wealth} and \cref{overlinex}, respectively, with the initial condition $(t_{0},x_{0})$ and the admissible controls $(\pi^{i},c^{i})_{i=1}^{n}$. For brevity, we denote $\mathcal{I}^{n}_{t_{0},x_{0}}$ as the set of all $n$-tuple of controls $(\pi^{i},c^{i})_{i=1}^{n}$ that satisfy the integrability condition over $[t_{0},T]$.
\end{Definition}
	According to the above definition, it is evident that $\mathcal{I}^{n}_{t_{0},x_{0}}\subset\mathcal{A}_{t_{0}}^{n}$. As we shall see later in \cref{dfadmissiable}, the outcome of a DF strategy always satisfies the integrability condition. Note that  $\mathcal{I}^{n}_{t_{0},x_{0}}$ is analogous to $\Pi^{x,p}_{2}\cap\Pi^{x,p/(p-1)}_{3}$ in \cite{2021Time}, with $p=2$ in our setting. Applying a similar proof as in  \cite{2021Time}, we can derive the following uniqueness result for the open-loop consistent control in the single-agent problem.
\begin{theorem}\label{uniquethem}
	Let $i\in\left\{1,\cdots,n\right\}$, $(t_{0},x_{0})\in[0,T)\times\mathbb{R}^{n}$ and $\left(\pi,c\right)^{(i)}=\left(\pi^{k},c^{k}\right)_{k\neq i}\in\mathcal{A}^{n-1}_{t_{0}}$ be given. Consider two admissible controls $(\pi^{i,1},c^{i,1})$ and $(\pi^{i,2},c^{i,2})$ that satisfy the following conditions:
	\begin{itemize}
	\item[(i)] $\left((\pi^{i,1},c^{i,1}),\left(\pi,c\right)^{(i)}\right)\in\mathcal{I}^{n}_{t_{0},x_{0}}$and $\left((\pi^{i,2},c^{i,2}),\left(\pi,c\right)^{(i)}\right)\in\mathcal{I}^{n}_{t_{0},x_{0}}$.
	\item[(ii)] Both $(\pi^{i,1},c^{i,1})$ and $(\pi^{i,2},c^{i,2})$ are open-loop consistent controls for agent $i$ with respect to the initial condition $(t_{0},x_{0})$ in response to $\left(\pi^{k},c^{k}\right)_{k\neq i}$.
	\end{itemize}  
Then, $(\pi^{i,1},c^{i,1})=(\pi^{i,2},c^{i,2})$.
\end{theorem}
\begin{proof} As 
the proof is exactly the same as the proof of  \cite[Theorem 5.3]{2021Time},  we omit it here.
\end{proof}
}

Now, we are ready for the following definitions.
\begin{Definition}[Open-loop equilibrium control]
	
Let {$(t_{0},x_{0})\in[0,T)\times\mathbb{R}^{n}$} be given. An $n$-tuple of admissible controls $(\pi^{*,i},c^{*,i})_{i=1}^{n}\in\mathcal{A}_{t_{0}}^{n}$  is said to be an open-loop equilibrium control with respect to the {initial condition $(t_{0},x_{0})$}  if for every $i\in\left\{1,\cdots,n\right\}$, $(\pi^{*,i},c^{*,i})$ is an open-loop consistent control for agent $i$ with respect to the {initial condition $(t_{0},x_{0})$} in response to $(\pi^{*,k},c^{*,k})_{k\neq i}$.
\end{Definition}
\begin{Definition} [DF equilibrium strategy]\label{DF}
	
	A DF strategy $\left(\Pi^{*,i}, C^{*,i}\right)_{i=1}^{n}\in\mathcal{S}^{n}$ is said to be a DF equilibrium strategy if for every {$(t_{0},x_{0})\in[0,T)\times\mathbb{R}^{n}$}, the corresponding outcome $\left(\pi^{*,i},c^{*,i}\right)_{i=1}^{n}\in\mathcal{A}^{n}_{t_{0}}$ is an open-loop equilibrium control with respect to the {initial condition $(t_{0},x_{0})$}. Moreover, a DF equilibrium strategy $\left(\Pi^{*,i}, C^{*,i}\right)_{i=1}^{n}$ is said to be simple if the DF strategy $\left(\Pi^{*,i}, C^{*,i}\right)_{i=1}^{n}$ is simple.
\end{Definition} 
 \begin{remark}
	 For the sake of brevity we use the name ``equilibrium strategy'' but it is different to ``close-loop equilibrium strategies'' or ``subgame-perfect equilibrium strategies'' discussed in the literature of time-inconsistent problems. A more appropriate name would be a ``DF representation of open-loop equilibrium controls''.
 \end{remark}
The main result of this section is the following, which gives the explicit form of a DF equilibrium strategy:

\begin{theorem}\label{The2.1}
	
	 Assume that for all $i=1,\cdots,n$ ($n\geq2$), we have $\delta_{i}>0$, $\theta_{i}\in[0,1)$, $\mu_{i}>0$, $\sigma_{i}\geq0$, $\nu_{i}\geq0$ and $\sigma_{i}+\nu_{i}>0$. Then there exists a unique simple  DF equilibrium strategy $(\Pi^{*,i},C^{*,i})_{i=1}^{n}$ being the following form:	
	\begin{align}
			&\Pi^{*,i}(t)=\left[\delta_{i}\frac{\mu_{i}}{\sigma_{i}^{2}+\left(1-\frac{\theta_{i}}{n}\right)\nu_{i}^{2}}+\theta_{i}\frac{\sigma_{i}}{\sigma_{i}^{2}+\left(1-\frac{\theta_{i}}{n}\right)\nu_{i}^{2}}\frac{\phi_{n}}{1-\psi_{n}}\right]\left(T+1-t\right), \label{2.5}\\
	        &C^{*,i}(t,x)=\frac{x^{i}}{T+1-t}-\delta_{i}\hat{h}^{i}(t)-\frac{\theta_{i}}{1-\overline{\theta}}\overline{\delta \hat{h}(t)}-\left(\delta_{i}+\theta_{i}\frac{\overline{\delta}}{1-\overline{\theta}}\right)\ln\left[\lambda(T-t)\right].\label{2.6}
	\end{align}
The constants $\phi_{n},\psi_{n},\overline{\delta}$ and $\overline{\theta}$ are 
\begin{equation}
	\begin{split}
		&\phi_{n}=\frac{1}{n}\sum_{k=1}^{n}\delta_{k}\frac{\sigma_{k}\mu_{k}}{\sigma_{k}^{2}+\left(1-\frac{\theta_{k}}{n}\right)\nu_{k}^{2}},\\
		&\psi_{n}=\frac{1}{n}\sum_{k=1}^{n}\theta_{k}\frac{\sigma_{k}^{2}}{\sigma_{k}^{2}+\left(1-\frac{\theta_{k}}{n}\right)\nu_{k}^{2}},\\
		&\overline{\delta}=\frac{1}{n}\sum_{k=1}^{n}\delta_{k},\quad
		\overline{\theta}=\frac{1}{n}\sum_{k=1}^{n}\theta_{k}.
	\end{split}\label{psi}
\end{equation}	
The function $(\hat{h}^{i}(\cdot),\overline{\delta \hat{h}(\cdot)})$ is
\begin{equation}\label{hhhh}
	\begin{split}
		\hat{h}^{i}(t)&=\frac{D^{i}_{n}}{2}\left[\frac{1}{T+1-t}-(T+1-t)\right]-\frac{1}{T+1-t}\int_{t}^{T}\ln\left[\lambda(T-s)\right]ds,\\
		\overline{\delta \hat{h}(t)}&=\frac{1}{n}\sum_{k=1}^{n}\delta_{k}\hat{h}^{k}(t),
	\end{split}
\end{equation}
where
\begin{equation}
	\begin{split}
    D_{n}^{i}&=\frac{1}{2}\frac{(\mu_{i}+\sigma_{i}A^{i}_{n})^{2}}{\nu_{i}^{2}+\sigma_{i}^{2}}-\frac{1}{2}\left[(A^{i}_{n})^{2}+C^{i}_{n}\right]-B^{i}_{n},\\
	A^{i}_{n}&=\frac{1}{n}\frac{\theta_{i}}{\delta_{i}}\sum_{k\ne i}\left[\delta_{k}\frac{\sigma_{k}\mu_{k}}{\sigma_{k}^{2}+\left(1-\frac{\theta_{k}}{n}\right)\nu_{k}^{2}}+\theta_{k}\frac{\sigma_{k}^{2}}{\sigma_{k}^{2}+\left(1-\frac{\theta_{k}}{n}\right)\nu_{k}^{2}}\frac{\phi_{n}}{1-\psi_{n}}\right],\\
	B^{i}_{n}&=\frac{1}{n}\frac{\theta_{i}}{\delta_{i}}\sum_{k\ne i}\left[\delta_{k}\frac{\mu_{k}^{2}}{\sigma_{k}^{2}+\left(1-\frac{\theta_{k}}{n}\right)\nu_{k}^{2}}+\theta_{k}\frac{\mu_{k}\sigma_{k}}{\sigma_{k}^{2}+\left(1-\frac{\theta_{k}}{n}\right)\nu_{k}^{2}}\frac{\phi_{n}}{1-\psi_{n}}\right],\\
	C^{i}_{n}&=(\frac{1}{n}\frac{\theta_{i}}{\delta_{i}})^{2}\sum_{k\ne i}\left[\delta_{k}\frac{\nu_{k}\mu_{k}}{\sigma_{k}^{2}+\left(1-\frac{\theta_{k}}{n}\right)\nu_{k}^{2}}+\theta_{k}\frac{\nu_{k}\sigma_{k}}{\sigma_{k}^{2}+\left(1-\frac{\theta_{k}}{n}\right)\nu_{k}^{2}}\frac{\phi_{n}}{1-\psi_{n}}\right]^{2}.
	\end{split}	
\end{equation}
\end{theorem}
It is straightforward to obtain the following corollary, which covers the single stock case.
\begin{Corollary}[Single stock]\label{2.2c}
	Assume that for all $i=1,\cdots,n$ $(n\geq2)$, we have $\mu_{i}=\mu>0$, $\sigma_{i}=\sigma>0$, and $\nu_{i}=0$. Then the strategy $(\Pi^{*,i},C^{*,i})_{i=1}^{n}$ has the following form:
   \begin{equation}
   	\begin{split}
   	 &\Pi^{*,i}(t)=\frac{\mu}{\sigma^{2}}\left(\delta_{i}+\theta_{i}\frac{\overline{\delta}}{1-\overline{\theta}}\right)\left(T+1-t\right),\\
	&C^{*,i}(t,x)=\frac{x^{i}}{T+1-t}+\left(\delta_{i}+\theta_{i}\frac{\overline{\delta}}{1-\overline{\theta}}\right)\mathcal{H}(t)-\left(\delta_{i}+\theta_{i}\frac{\overline{\delta}}{1-\overline{\theta}}\right)\ln\left[\lambda(T-t)\right],	
		\end{split}
	\end{equation}
where the function $\mathcal{H}(\cdot)$ is 
	\begin{equation}\label{fH}
		\mathcal{H}(t)=(\frac{1}{2}\frac{\mu}{\sigma})^{2}\left[(T+1-t)-\frac{1}{T+1-t}\right]+\frac{1}{T+1-t}\int_{t}^{T}\ln\left[\lambda(T-s)\right]ds.
	\end{equation}
\end{Corollary}		
\begin{proof}[Proof of  \cref{The2.1}]\label{proof828}
 Let $i\in\left\{1,\cdots,n\right\}$ and {$(t_{0},x_{0})\in[0,T)\times\mathbb{R}^{n}$} be given. We denote by $(\hat{\pi}^{i},\hat{c}^{i})\in{\mathcal{A}_{t_{0}}}$ a candidate control for agent $i$. Assume that the inputs $(\pi^{k},c^{k})_{k\neq i}$ are of the following form:
	\begin{equation*}
		\left(\pi^{k}_{t},c^{k}_{t}\right)=\left(\Pi^{k}(t),C^{k}(t,X_{t})\right),\,\ k\neq i,\,\ {t\in[t_{0},T]},
	\end{equation*}
	where  $\left(\Pi^{k},C^{k}\right)_{k\neq i}\in\mathcal{S}^{n-1}$ and  $X=(X^{1},\cdots,X^{n})$ is the wealth process associated with $\left(\hat{\pi}^{i},\hat{c}^{i}\right)$ and $(\pi^{k},c^{k})_{k\neq i}$ satisfying
	\begin{equation}
		\begin{cases}
			dX^{k}_{t}=\left(\Pi^{k}(t)\mu_{k}-C^{k}(t,X_{t})\right)dt+\Pi^{k}(t)\nu_{k}dW^{k}_{t}+\Pi^{k}(t)\sigma_{k}dB_{t},\, k\neq i,\, {t\in[t_{0},T]}, \\
			dX^{i}_{t}=\left(\hat{\pi}_{t}^{i}\mu_{i}-\hat{c}^{i}_{t}\right)dt+\hat{\pi}_{t}^{i}\nu_{i}dW_{t}^{i}+\hat{\pi}_{t}^{i}\sigma _{i}dB_{t},\, {t\in[t_{0},T]},\\ X^{j}_{{t_{0}}}=x^{j}_{0},\, j=1,\dots,n.
		\end{cases}
	\end{equation}

	  We now find the open-loop consistent control for agent $i$ with respect to the { initial condition $(t_{0},x_{0})$} in response to $\left(\pi^{k},c^{k}\right)_{k\neq i}$, then we resolve the resulting fixed point problem to obtain the desired equilibrium.

\textit{Open-loop consistent control for agent $i$}.  According to  \cref{verif} in Appendix A,  we first aim to find  a classical solution of the following PDE:
\begin{equation}
\begin{cases}
	V^{i}_{t}(t,x^{i},\overline{x}^{(i)})+\hat{\Pi}^{i}(t,x)\left(V^{i}_{x}(t,x^{i},\overline{x}^{(i)})\mu_{i}+V^{i}_{xy}(t,x^{i},\overline{x}^{(i)})\sigma_{i}\overline{\sigma\Pi}^{(i)}(t)\right)\\+\frac{1}{2}\hat{\Pi}^{i}(t,x)^{2}V^{i}_{xx}(t,x^{i},\overline{x}^{(i)})\left(\nu_{i}^{2}+\sigma_{i}^{2}\right)-\hat{C}^{i}(t,x)V^{i}_{x}(t,x^{i},\overline{x}^{(i)})\\+V^{i}_{y}(t,x^{i},\overline{x}^{(i)})\left(\overline{\mu\Pi}^{(i)}(t)-\overline{C}^{(i)}(t,x)\right)\\+\frac{1}{2}V^{i}_{yy}(t,x^{i},\overline{x}^{(i)})\left[\left(\overline{\sigma\Pi}^{(i)}(t)\right)^{2}+\frac{1}{n^{2}}\sum_{k\neq i}\left(\nu_{k}\Pi^{k}(t)\right)^{2}\right]=0, (t,x)\in[0,T]\times\mathbb{R}^{n},\\
	V^{i}(T,x^{i},\overline{x}^{(i)})=\frac{1}{\delta_{i}}\left(1-\frac{\theta_{i}}{n}\right)\exp\left\{-\frac{1}{\delta_{i}}\left(1-\frac{\theta_{i}}{n}\right)x^{i}+\frac{\theta_{i}}{\delta_{i}}\overline{x}^{(i)}\right\},\,x\in\mathbb{R}^{n},\\
	\hat{\Pi}^{i}(t,x)=-\frac{\mu_{i}V^{i}(t,x^{i},\overline{x}^{(i)})+\sigma_{i}\overline{\sigma\Pi}^{(i)}(t)V^{i}_{y}(t,x^{i},\overline{x}^{(i)})}{(\nu_{i}^{2}+\sigma_{i}^{2})V^{i}_{x}(t,x^{i},\overline{x}^{(i)})},\,(t,x)\in[0,T]\times\mathbb{R}^{n},\\
	\hat{C}^{i}(t,x)=-\frac{\delta_{i}}{1-\frac{\theta_{i}}{n}}\ln\left[\frac{\delta_{i}}{1-\frac{\theta_{i}}{n}}\lambda(T-t)V^{i}(t,x^{i},\overline{x}^{(i)})\right]\\+\frac{\theta_{i}}{1-\frac{\theta_{i}}{n}}\overline{C}^{(i)}(t,x),\,(t,x)\in[0,T]\times\mathbb{R}^{n}.
\end{cases}	\label{sys1}
\end{equation}
Based on the terminal condition, we guess that the solution has the following form:
\begin{equation*}
	V^{i}(t,x^{i},\overline{x}^{(i)})=\frac{1}{\delta_{i}}\left(1-\frac{\theta_{i}}{n}\right)\exp\left\{f^{i}(t)x^{i}+g^{i}(t)\overline{x}^{(i)}+h^{i}(t)\right\}, \  \left(t,x^{i},\overline{x}^{(i)}\right)\in[0,T]\times\mathbb{R}\times\mathbb{R},\label{Guass}
\end{equation*} 
where $f^{i}(\cdot)$, $g^{i}(\cdot)$ and $h^{i}(\cdot)\in C^{1}([0,T];\mathbb{R})$ such that $f^{i}(T)=-\frac{1}{\delta_{i}}\left(1-\frac{\theta_{i}}{n}\right)$, $g^{i}(T)=\frac{\theta_{i}}{\delta_{i}}$ and $h^{i}(T)=0$.

Substituting the ansatz into \cref{sys1}, we obtain
\begin{align*}
	&\left[\dot{f}^{i}(t)+\frac{\delta_{i}}{1-\frac{\theta_{i}}{n}}\left(f^{i}(t)\right)^{2}\right]x^{i}+\left[\dot{g}^{i}(t)+\frac{\delta_{i}}{1-\frac{\theta_{i}}{n}}f^{i}(t)g^{i}(t)\right]\overline{x}^{(i)}\\&+\dot{h}^{i}(t)+\frac{\delta_{i}}{1-\frac{\theta_{i}}{n}}f^{i}(t)h^{i}(t)+\frac{\delta_{i}}{1-\frac{\theta_{i}}{n}}f^{i}(t)\ln\left[\lambda(T-t)\right]\\&-\frac{1}{2}\frac{\left(\mu_{i}+\sigma_{i}\overline{\sigma\Pi}^{(i)}(t)g^{i}(t)\right)^{2}}{\nu_{i}^{2}+\sigma_{i}^{2}}+\overline{\mu\Pi}^{(i)}(t)g^{i}(t)+\frac{1}{2}\left[g^{i}(t)\right]^{2}\left[\left(\overline{\sigma\Pi}^{(i)}(t)\right)^{2}+\frac{1}{n^{2}}\sum_{k\neq i}\left(\nu_{k}\Pi^{k}(t)\right)^{2}\right]\\&-\left(\frac{\theta_{i}}{1-\frac{\theta_{i}}{n}}f^{i}(t)+g^{i}(t)\right)\overline{C}^{(i)}(t,x)=0,\,\ (t,x)\in[0,T]\times\mathbb{R}^{n}.
\end{align*}

Although the term $\overline{C}^{(i)}(t,x)$ in the last line depends on $x\in\mathbb{R}^{n}$, we can solve the above equation correctly. Indeed, if $f^{i}$ and $g^{i}$ solve the following ODEs:
\begin{equation}\label{ODE1}
	\begin{cases}
		\dot{f}^{i}(t)+\frac{\delta_{i}}{1-\frac{\theta_{i}}{n}}\left(f^{i}(t)\right)^{2}=0,\quad t\in[0,T],\\
		\dot{g}^{i}(t)+\frac{\delta_{i}}{1-\frac{\theta_{i}}{n}}f^{i}(t)g^{i}(t)=0,\quad t\in[0,T],\\
		f^{i}(T)=-\frac{1}{\delta_{i}}\left(1-\frac{\theta_{i}}{n}\right),\, g^{i}(T)=\frac{\theta_{i}}{\delta_{i}},
	\end{cases}
\end{equation}
whose unique solutions are given by
\begin{equation}
\begin{cases}
	    f^{i}(t)=-\frac{1}{\delta_{i}}\left(1-\frac{\theta_{i}}{n}\right)\frac{1}{T+1-t},\\
	    g^{i}(t)=\frac{\theta_{i}}{\delta_{i}}\frac{1}{T+1-t},
\end{cases}
\end{equation}
then it holds that $\frac{\theta_{i}}{1-\frac{\theta_{i}}{n}}f^{i}(t)+g^{i}(t)=0$ for any $t\in[0,T]$. Thus, the coefficient of $\overline{C}^{(i)}(t,x)$ vanishes, and we obtain the ODE of $h^{i}$:
\begin{equation*}
	\begin{split}
		&\dot{h}^{i}(t)+\frac{\delta_{i}}{1-\frac{\theta_{i}}{n}}f^{i}(t)h^{i}(t)+\frac{\delta_{i}}{1-\frac{\theta_{i}}{n}}f^{i}(t)\ln\left[\lambda(T-t)\right]\\&-\frac{1}{2}\frac{\left(\mu_{i}+\sigma_{i}\overline{\sigma\Pi}^{(i)}(t)g^{i}(t)\right)^{2}}{\nu_{i}^{2}+\sigma_{i}^{2}}+\overline{\mu\Pi}^{(i)}(t)g^{i}(t)+\frac{1}{2}\left[g^{i}(t)\right]^{2}\left[\left(\overline{\sigma\Pi}^{(i)}(t)\right)^{2}+\frac{1}{n^{2}}\sum_{k\neq i}\left(\nu_{k}\Pi^{k}(t)\right)^{2}\right]\\&=0
	\end{split}
\end{equation*}
with terminal condition $h^{i}(T)=0$, whose solution is given by
\begin{equation} h^{i}(t)=\frac{1}{T+1-t}\int_{t}^{T}\left(T+1-s\right)G^{i}((\Pi^{k}(s))_{k\ne i},s)ds,
\end{equation}
where the deterministic function $G^{i}:\mathbb{R}^{n-1}\times[0,T]\rightarrow\mathbb{R}$ is defined by  \begin{align*}
    G^{i}((\pi^{k})_{k\ne i},t):&=-\frac{1}{T+1-t}\ln\left[\lambda(T-t)\right]-\frac{1}{2}\frac{\left(\mu_{i}+\frac{\theta_{i}\sigma_{i}}{\left(T+1-t\right)\delta_{i}}\overline{\sigma\pi}^{(i)}\right)^{2}}{\nu_{i}^{2}+\sigma_{i}^{2}}\\&+\frac{\theta_{i}}{\left(T+1-t\right)\delta_{i}}\overline{\mu\pi}^{(i)}+\frac{\theta_{i}^{2}}{2(T+1-t)^{2}\delta_{i}^{2}}\left[(\overline{\sigma\pi}^{(i)})^{2}+\frac{1}{n^{2}}\sum_{k\neq i}(\nu_{k}\pi^{k})^{2}\right]
\end{align*}
for $((\pi^{k})_{k\ne i},t)\in\mathbb{R}^{n-1}\times[0,T]$.

Then, by the representation of $(\hat{\Pi}^{i},\hat{C}^{i})$ in \cref{sys1}, we have
\begin{equation}\label{candidateweak}
	\begin{cases}
	\hat{\Pi}^{i}(t,x)=\frac{\mu_{i}+\sigma_{i}\overline{\sigma\Pi}^{(i)}(t)\frac{\theta_{i}}{\delta_{i}}\frac{1}{T+1-t}}{\left(\nu_{i}^{2}+\sigma_{i}^{2}\right)\left[\frac{1}{\delta_{i}}(1-\frac{\theta_{i}}{n})\frac{1}{T+1-t}\right]},\, (t,x)\in[0,T]\times\mathbb{R}^{n},\\
	\hat{C}^{i}(t,x)=\frac{1}{T+1-t}x^{i}-\frac{\theta_{i}}{1-\frac{\theta_{i}}{n}}\frac{1}{T+1-t}\overline{x}^{(i)}-\frac{\delta_{i}}{1-\frac{\theta_{i}}{n}}h^{i}(t)-\frac{\delta_{i}}{1-\frac{\theta_{i}}{n}}\ln\left[\lambda(T-t)\right]\\+\frac{\theta_{i}}{1-\frac{\theta_{i}}{n}}\overline{C}^{(i)}(t,x),\, (t,x)\in[0,T]\times\mathbb{R}^{n}.
	\end{cases}
\end{equation}

In particular, $\hat{\Pi}^{i}(t,x)=	\hat{\Pi}^{i}(t)$ does not depend on the state argument $x\in\mathbb{R}^{n}$. It is clear that $(\hat{\Pi}^{i},\hat{C}^{i})\in\mathcal{S}$ and does not depend on the {initial condition $(t_{0},x_{0})$}. Using \cref{dfadmissiable} and  \cref{verif},  we can conclude that  the outcome $(\hat{\pi}^{i},\hat{c}^{i})$ associated with $(\hat{\Pi}^{i},\hat{C}^{i})\in\mathcal{S}$ is indeed an open-loop consistent control for agent $i$ with respect to the {initial condition $(t_{0},x_{0})$} in response to $\left(\pi^{k},c^{k}\right)_{k\neq i}$. 

\textit{Fixed point problem}. From the above discussion, we actually construct a best response map
\begin{equation}
	F:\left(\Pi^{k},C^{k}\right)_{k=1}^{n}\rightarrow\left(\hat{\Pi}^{k},\hat{C}^{k}\right)_{k=1}^{n}
\end{equation}
by \cref{candidateweak}, which maps a DF strategy into a DF strategy. Now, we aim to find the fixed point of  the map $F$.

 We first address the investment strategies. For a candidate portfolio vector $(\Pi^{1},\dots,\Pi^{n})$ to be the fixed point, we need $\Pi^{i}(t)=\hat{\Pi}^{i}(t)$, for $i=1,\dots,n$ and $t\in[0,T]$. Let\[\overline{\sigma\Pi}(t):=\frac{1}{n}\sum_{k=1}^{n}\sigma_{k}\Pi^{k}(t)=\overline{\sigma\Pi}^{(i)}(t)+\frac{1}{n}\sigma_{i}\Pi^{i}(t).\]

Then, \[\Pi^{i}(t)=\hat{\Pi}^{i}(t)=\frac{\delta_{i}\mu_{i}+\theta_{i}\sigma_{i}\overline{\sigma\Pi}(t)\frac{1}{T+1-t}}{\left(\nu_{i}^{2}+\sigma_{i}^{2}\right)\left(1-\frac{\theta_{i}}{n}\right)\frac{1}{T+1-t}}-\frac{\theta_{i}\sigma_{i}^{2}\Pi^{i}(t)}{n\left(\nu_{i}^{2}+\sigma_{i}^{2}\right)\left(1-\frac{\theta_{i}}{n}\right)},\]
which yields
\begin{equation}\label{piaveg}
	\begin{split}
		\Pi^{i}(t)&=\frac{\delta_{i}\mu_{i}+\theta_{i}\sigma_{i}\overline{\sigma\Pi}(t)\frac{1}{T+1-t}}{\left(\nu_{i}^{2}+\sigma_{i}^{2}\right)\left(1-\frac{\theta_{i}}{n}\right)\frac{1}{T+1-t}}\left(1+\frac{\theta_{i}\sigma_{i}^{2}}{n\left(\nu_{i}^{2}+\sigma_{i}^{2}\right)\left(1-\frac{\theta_{i}}{n}\right)}\right)^{-1}\\&=\frac{\delta_{i}\mu_{i}+\theta_{i}\sigma_{i}\overline{\sigma\Pi}(t)\frac{1}{T+1-t}}{\left[\sigma_{i}^{2}+(1-\frac{\theta_{i}}{n})\nu_{i}^{2}\right]\frac{1}{T+1-t}}.
	\end{split}
\end{equation}

Multiplying both sides of \cref{piaveg}  by $\sigma_{i}$ and then averaging over $i=1,\dots,n$, gives the following fixed point equation:
\begin{equation}
	\overline{\sigma\Pi}(t)=\phi_{n}(T+1-t)+\psi_{n}\overline{\sigma\Pi}(t)\label{ep}
\end{equation}
where $(\phi_{n},\psi_{n})$ is as defined in \cref{psi}.

We then have the following cases to get the fixed point:
\begin{itemize}
\item [(i)]  If $\psi_{n}<1$, then \cref{ep}  yields $\overline{\sigma\Pi}(t)=\left[\phi_{n}/\left(1-\psi_{n}\right)\right]\left(T+1-t\right)$ and the  investment strategy is given by \cref{2.5}.

\item [(ii)]  If $\psi_{n}=1$, then the equation \cref{ep} has no solution. Note that $\psi_{n}=1$ and $\phi_{n}=0$ cannot happen. Using assumption $\delta_{i}>0$, $\mu_{i}>0$, and $\sigma_{i}+\nu_{i}>0$, one can easily get a contradiction.
\end{itemize}
Next, we address the consumption strategies. Similarly, in order to be the fixed point, the candidate consumption vector $(C^{1},\dots,C^{n})$ needs to satisfy $C^{i}(t,x)=\hat{C}^{i}(t,x)$, for $i=1,\dots,n$ and $(t,x)\in[0,T]\times\mathbb{R}^{n}$. Let \[\overline{C}(t,x):=\frac{1}{n}\sum_{k=1}^{n}C^{k}(t,x)=\overline{C}^{(i)}(t,x)+\frac{1}{n}C^{i}(t,x).\]

Then, 
\begin{equation}
\begin{split}
    C^{i}(t,x)&=\frac{1}{T+1-t}x^{i}-\frac{\theta_{i}}{1-\frac{\theta_{i}}{n}}\frac{1}{T+1-t}\overline{x}^{(i)}-\frac{\delta_{i}}{1-\frac{\theta_{i}}{n}}\hat{h}^{i}(t)+\frac{\theta_{i}}{1-\frac{\theta_{i}}{n}}\overline{C}(t,x)\\&-\frac{\theta_{i}}{1-\frac{\theta_{i}}{n}}\frac{C^{i}(t,x)}{n}-\frac{\delta_{i}}{1-\frac{\theta_{i}}{n}}\ln\left[\lambda(T-t)\right],\label{c1}
\end{split}
\end{equation}
 where $\hat{h}^{i}(t):=\frac{1}{T+1-t}\int_{t}^{T}(T+1-s)\left[G^{i}((\Pi^{*,k}(s))_{k\ne i},s)\right]ds$. For explicit expression of $\hat{h}^{i}(t)$ in \cref{hhhh}, see \cref{h}.  
 
 The  equation \cref{c1} implies 
\begin{equation}\label{different}
	C^{i}(t,x)=\frac{1}{T+1-t}\left(x^{i}-\theta_{i}\overline{x}\right)-\delta_{i}\hat{h}^{i}(t)+\theta_{i}\overline{C}(t,x)-\delta_{i}\ln\left[\lambda(T-t)\right].
\end{equation}

Averaging over $i=1,\dots,n$, gives 
\begin{equation}
	\left(1-\overline{\theta}\right)\overline{C}(t,x)=\frac{\left(1-\overline{\theta}\right)\overline{x}}{T+1-t}-\overline{\delta \hat{h}(t)}-\overline{\delta}\ln\left[\lambda(T-t)\right],\label{eqc}
\end{equation}
where $\overline{(\cdot)}$ represents the arithmetic mean.

Then, we have the following cases to get the fixed point:
\begin{itemize}
\item [(i)] If $\overline{\theta}<1$, then equation \cref{eqc} yields $\overline{C}(t,x)=\frac{\overline{x}}{T+1-t}-\frac{\overline{\delta \hat{h}(t)}}{1-\overline{\theta}}-\frac{\overline{\delta}}{1-\overline{\theta}}\ln\left[\lambda(T-t)\right]$, and the consumption strategy is given by \cref{2.6}.

\item [(ii)] If $\overline{\theta}=1$ and $\overline{\delta \hat{h}(t)}+\overline{\delta}\ln\left[\lambda(T-t)\right]\not\equiv0$, then equation \cref{eqc} has no solution.

\item [(iii)] If $\overline{\theta}=1$ and $\overline{\delta \hat{h}(t)}+\overline{\delta}\ln\left[\lambda(T-t)\right]\equiv0$, then there exist infinitely many solutions.
\end{itemize}

In order to remove `bad' cases, we limit $\theta_{i}$ in $[0,1)$ for all $i\in\left\{1,\cdots,n\right\}$. In summary, we get that there exists a unique solution to the fixed point problem, which turns out to be a simple DF equilibrium strategy $(\Pi^{*,i},C^{*,i})_{i=1}^{n}$  given by \cref{2.5} and \cref{2.6}.

\textit{Uniqueness.}  We claim that a simple DF equilibrium strategy is equivalent to  a fixed point of the map $F$. Indeed, based on the above discussion, we conclude that the fixed point of $F$ must be a simple DF equilibrium strategy. {To complete the argument, we need to demonstrate that a simple DF equilibrium strategy must be a fixed point of $F$.

{Assume that $\left(\Pi^{*,k},C^{*,k}\right)_{k=1}^{n}$ is a simple DF equilibrium strategy. Define $$\left(\hat{\Pi}^{k},\hat{C}^{k}\right)_{k=1}^{n}:=F\left(\left(\Pi^{*,k},C^{*,k}\right)_{k=1}^{n}\right).$$  Take an arbitrary initial pair $(t_{0},x_{0})\in[0,T)\times\mathbb{R}^{n}$ and fix $i\in\left\{1,\cdots,n\right\}$. Then,  we consider the controls $\left(\pi^{*,k},c^{*,k}\right)_{k=1}^{n}\in\mathcal{A}^{n}_{t_{0}}$ and $\left(\hat{\pi}^{[i],k},\hat{c}^{[i],k}\right)_{k=1}^{n}\in\mathcal{A}^{n}_{t_{0}}$ determined by 
\begin{equation*}
	\begin{cases}
		\left(\pi^{*,i}_{t},c^{*,i}_{t}\right):=\left(\Pi^{*,i}(t),C^{*,i}(t,X^{*}_{t})\right),\,\, t\in[t_{0},T],\\
		\left(\pi^{*,k}_{t},c^{*,k}_{t}\right):=\left(\Pi^{*,k}(t),C^{*,k}(t,X^{*}_{t})\right),\,\, t\in[t_{0},T],\, k\neq i,
	\end{cases}
\end{equation*}}
{\begin{equation*}
   \begin{cases}
	\left(\hat{\pi}^{[i],i}_{t},\hat{c}^{[i],i}_{t}\right):=\left(\hat{\Pi}^{i}(t),\hat{C}^{i}(t,\hat{X}^{[i]}_{t})\right),\,\, t\in[t_{0},T],\\
	\left(\hat{\pi}^{[i],k}_{t},\hat{c}^{[i],k}_{t}\right):=\left(\Pi^{*,k}(t),C^{*,k}(t,\hat{X}^{[i]}_{t})\right),\,\, t\in[t_{0},T],\, k\neq i,
   \end{cases}
\end{equation*}
where $X^{*}=\left(X^{*,1},\cdots,X^{*,n}\right)$ and $\hat{X}^{[i]}=\left(\hat{X}^{[i],1},\cdots,\hat{X}^{[i],n}\right)$ are defined by}
{\begin{equation*}
	\begin{cases}
		dX^{*,i}_{t}=\left(\Pi^{*,i}(t)\mu_{i}-C^{*,i}(t,X^{*}_{t})\right)dt+\Pi^{*,i}(t)\nu_{i}dW^{i}_{t}+\Pi^{*,i}(t)\sigma_{i}dB_{t},\,\ t\in[t_{0},T],\\
		dX^{*,k}_{t}=\left(\Pi^{*,k}(t)\mu_{k}-C^{*,k}(t,X^{*}_{t})\right)dt+\Pi^{*,k}(t)\nu_{k}dW^{k}_{t}+\Pi^{*,k}(t)\sigma_{k}dB_{t},\,\ t\in[t_{0},T],\,\\ k\neq i,\\
		X^{*}_{t_{0}}=x_{0}\in\mathbb{R}^{n},
	\end{cases}
\end{equation*}
\begin{equation*}
	\begin{cases}
		d\hat{X}^{[i],i}_{t}=\left(\hat{\Pi}^{i}(t)\mu_{i}-\hat{C}^{i}(t,\hat{X}^{[i]}_{t})\right)dt+\hat{\Pi}^{i}(t)\nu_{i}dW^{i}_{t}+\hat{\Pi}^{i}(t)\sigma_{i}dB_{t},\,\ t\in[t_{0},T],\\
		d\hat{X}^{[i],k}_{t}=\left(\Pi^{*,k}(t)\mu_{k}-C^{*,k}(t,\hat{X}^{[i]}_{t})\right)dt+\Pi^{*,k}(t)\nu_{k}dW^{k}_{t}+\Pi^{*,k}(t)\sigma_{k}dB_{t},\,\ t\in[t_{0},T],\,\\ k\neq i,\\
		\hat{X}^{[i]}_{t_{0}}=x_{0}\in\mathbb{R}^{n}.
	\end{cases}
\end{equation*}}
Since $\left(\Pi^{*,k},C^{*,k}\right)_{k=1}^{n}$ is simple, we see that $X^{*,k}$ and $\hat{X}^{[i],k}$ solve the same SDE for each $k\neq i$, and we get $X^{*,k}=\hat{X}^{[i],k}$ by the uniqueness of the SDE. Again by the assumption that $\left(\Pi^{*,k},C^{*,k}\right)_{k=1}^{n}$ is simple, we see that $\left(\pi^{*,k},c^{*,k}\right)_{k\neq i}=\left(\hat{\pi}^{[i],k},\hat{c}^{[i],k}\right)_{k\neq i}$. 
By the construction of controls $\left(\pi^{*,k},c^{*,k}\right)_{k=1}^{n}$ and $\left(\hat{\pi}^{[i],k},\hat{c}^{[i],k}\right)_{k=1}^{n}$, we have that 
\begin{itemize}
	{\item [ (i)] $\left(\pi^{*,i},c^{*,i}\right)$ is an  open-loop consistent control for agent $i$ with respect to the initial condition $(t_{0},x_{0})$ in response to $\left(\pi^{*,k},c^{*,k}\right)_{k\neq i}$;}
	{\item [ (ii)] $\left(\hat{\pi}^{[i],i},\hat{c}^{[i],i}\right)$ is an open-loop consistent control for agent $i$ with respect to the initial condition $(t_{0},x_{0})$ in response to $\left(\pi^{[i],k},c^{[i],k}\right)_{k\neq i}$.}
\end{itemize}
As both $\left(\pi^{*,k},c^{*,k}\right)_{k=1}^{n}$ and $\left(\hat{\pi}^{[i],k},\hat{c}^{[i],k}\right)_{k=1}^{n}$ are outcomes of DF strategies,  they are in $\mathcal{I}^{n}_{t_{0},x_{0}}$ by \cref{dfadmissiable}. Then, we deduce that  $\left(\pi^{*,i},c^{*,i}\right)=\left(\hat{\pi}^{[i],i},\hat{c}^{[i],i}\right)$ by the uniqueness of the open-loop consistent control (with respect to the initial condition $(t_{0},x_{0})$), see \cref{uniquethem}, and thus we have $X^{*}=\hat{X}^{[i]}$.  Therefore, we have $\Pi^{*,i}(t)=\hat{\Pi}^{i}(t)$ and $C^{*,i}(t,X^{*}_{t})=\hat{C}^{i}(t,X^{*}_{t})$ for any $t\in[t_{0},T]$. In particular, we have $\Pi^{*,i}(t_{0})=\hat{\Pi}^{i}(t_{0})$ and $C^{*,i}(t_{0},x_{0})=\hat{C}^{i}(t_{0},x_{0})$. Since $(t_{0},x_{0})\in[0,T)\times\mathbb{R}^{n}$ is arbitrary, we see that $\Pi^{*,i}=\hat{\Pi}^{i}$ and $C^{*,i}=\hat{C}^{i}$.  As $i\in\left\{1,\cdots,n\right\}$ is arbitrary, we see that  $\left(\Pi^{*,k},C^{*,k}\right)_{k=1}^{n}$ is a fixed point of the map $F:\mathcal{S}^{n}\rightarrow\mathcal{S}^{n}$. Using the fact that  $F$ has a unique fixed point, we complete the proof.}
\end{proof}
\begin{remark}

Let us revisit the case of infinitely many solutions to the fixed point problem. The discount function can be obtained by solving the following equations,
	\begin{equation*}
		\begin{cases}
			\overline{\delta \hat{h}(t)}+\overline{\delta}\ln\left[\lambda(T-t)\right]=0,\,\ t\in[0,T],\\
			\lambda(0)=1.
		\end{cases}
    \end{equation*}
By calculation, we get
\begin{equation*}
	\begin{cases}
		\overline{\delta}\alpha(t)=\frac{\overline{\delta}}{T+1-t}\int_{t}^{T}\alpha(s)ds+\frac{\overline{\delta D_{n}}}{2}\left[(T+1-t)-\frac{1}{T+1-t}\right],\,\ t\in[0,T],\\
		\alpha(T)=0,
	\end{cases}
\end{equation*}
	which is equivalent to the following ODE,
	\begin{equation*}
	\begin{cases}
		\alpha'(t)=-\frac{\overline{\delta D_{n}}}{\overline{\delta}},\,\ t\in[0,T],\\
		\alpha(T)=0,
	\end{cases}
    \end{equation*}
    where $\alpha(t):=\ln\left[\lambda(T-t)\right]$, $\overline{\delta D_{n}}:=\frac{1}{n}\sum_{k=1}^{n}\delta_{k}D_{n}^{k}$ and $\overline{\delta}:=\frac{1}{n}\sum_{k=1}^{n}\delta_{k}$. Then, we can solve that the discount function is $\lambda(T-t)=\exp\{\frac{\overline{\delta D_{n}}}{\overline{\delta}}(T-t)\}$. 
\end{remark}
\section{The mean field game}\label{sect3}

In this section, we investigate the limit as $n\rightarrow\infty$ of the $n$-agent game discussed in Section \ref{sect2}. 

First we will use a heuristic argument as \cite{2020Many, D2017Mean} do to build intuition.
For the $n$-agent games, we define, for each agent $i=1,\cdots,n$, the type vector
\[\xi^{i}:=(\delta_{i},\theta_{i},\mu_{i},\nu_{i},\sigma_{i}).\]
These type vectors induce an empirical measure, called the type distribution, which is the probability measure on the type space:
\[\mathcal{Z}:=(0,\infty)\times[0,1)\times(0,\infty)\times[0,\infty)\times[0,\infty),\]
given by \[m_{n}(A)=\frac{1}{n}\sum_{i=1}^{n}1_{A}(\xi^{i}),\,\text{for Borel sets}\ A\subset\mathcal{Z}.\]

 	We assume that $m_{n}$ converges to some limiting probability measure $m$. To pass to the limit, let us denote a $\mathcal{Z}$-valued random variable $\xi=(\delta,\theta,\mu,\nu,\sigma)$ with distribution $m$. We should expect the strategy $(\Pi^{*,i},C^{*,i})$ to converge to
\begin{equation}
	\lim_{n\rightarrow\infty}\Pi^{*,i}(t)=\left[\delta_{i}\frac{\mu_{i}}{\sigma_{i}^{2}+\nu_{i}^{2}}+\theta_{i}\frac{\sigma_{i}}{\sigma_{i}^{2}+\nu_{i}^{2}}\frac{\phi}{1-\psi}\right]\left(T+1-t\right),\label{limitpi}
\end{equation}
where \[{\psi}:=\mathbb{E}\left[\theta\frac{\sigma^{2}}{\sigma^{2}+\nu^{2}}\right]  \quad\text{and}\quad\phi:=\mathbb{E}\left[\delta\frac{\mu\sigma}{\sigma^{2}+\nu^{2}}\right], \]
and
\begin{equation}
	\lim_{n\rightarrow\infty}C^{*,i}(t,x)=\frac{x^{i}}{T+1-t}-\delta_{i}H^{i}(t)-\theta_{i}\frac{\mathbb{E}\left[\delta H^{\xi}(t) \right]}{1-\mathbb{E}\left[\theta\right]}-\left[\delta_{i}+\theta_{i}\frac{\mathbb{E}[\delta ]}{1-\mathbb{E}[\theta]}\right]\ln\left[\lambda(T-t)\right],\label{limitc}
\end{equation}
where  $H^{i}(\cdot)$ is 
 \begin{equation}
	H^{i}(t):=\lim_{n\rightarrow\infty}\hat{h}^{i}(t)=\frac{D^{i}}{2}\left[\frac{1}{T+1-t}-(T+1-t)\right]-\frac{1}{T+1-t}\int_{t}^{T}\ln\left[\lambda(T-s)\right]ds,\label{H}
\end{equation}
with
\begin{equation}
	\begin{split}
	D^{i}:&=\frac{1}{2}\frac{\left(\mu_{i}+\sigma_{i}A^{i}\right)^{2}}{\nu_{i}^{2}+\sigma_{i}^{2}}-\frac{1}{2}\left(A^{i}\right)^{2}-B^{i},\\
		A^{i}:&=\lim_{n\rightarrow\infty} A^{i}_{n}=\frac{\theta_{i}}{\delta_{i}}\mathbb{E}\left[\delta\frac{\sigma\mu}{\sigma^{2}+\nu^{2}}+\theta\frac{\sigma^{2}}{\sigma^{2}+\nu^{2}}\frac{\phi}{1-\psi}\right],\\
		B^{i}:&=\lim_{n\rightarrow\infty}B^{i}_{n}=\frac{\theta_{i}}{\delta_{i}}\mathbb{E}\left[\delta\frac{\mu^{2}}{\sigma^{2}+\nu^{2}}+\theta\frac{\sigma\mu}{\sigma^{2}+\nu^{2}}\frac{\phi}{1-\psi}\right],
	\end{split}
\end{equation}
and $H^{\xi}$ is denoted as its randomization.  In this section, we always assume $\mathbb{E}\left[f(\xi)\right]<\infty$, where $f:\mathcal{Z}\rightarrow\mathbb{R}$ is an arbitrary Borel measurable function. We explain that this assumption is needed to construct the best response map. If one only wants to verify that the obtained strategy  is a DF equilibrium strategy, it is enough to assume all the expectations in the aforementioned discussion are finite.

We next explain how this strategy arises as the equilibrium of the MFG. Now, we assume that the filtered probability space $(\Omega,\mathcal{F},\mathbb{F}^{MF},\mathbb{P})$ supports  independent Brownian motions $W$ and $B$, as well as random type vector
$\xi=(\delta,\theta,\mu,\nu,\sigma)$
independent of $W$ and $B$, and with values in the space $\mathcal{Z}$, where $\mathbb{F}^{MF}:=(\mathcal{F}_{t}^{MF})_{t\in[0,T]}$ is the minimal filtration satisfying the usual assumptions such that $\xi$ is $\mathcal{F}_{0}^{MF}$-measurable and both $W$ and $B$ are $\mathbb{F}^{MF}$- Brownian motions. Let also 
$\mathbb{F}^{B}:=(\mathcal{F}^{B}_{t})_{t\in[0,T]}$ denote the natural filtration generated by the Brownian motion $B$.

Then the representative agent's wealth process is determined by 
\begin{equation}
	dX^{\xi}_{t}=\pi_{t}\left(\mu dt+\nu dW_{t}+\sigma dB_{t}\right)-c_{t}dt,\quad
	X^{\xi}_{t_{0}}=x_{0}\in\mathbb{R}.\label{mse}
\end{equation}

As before, we define admissible control and DF strategy in the framework of MFG.
\begin{Definition}[Admissible control]

A control $(\pi,c)$ is said to be admissible over $[t,T]$, if $(\pi,c)$ is an $\mathbb{F}^{MF}$-progressively measurable process and for any given deterministic sample $\xi_{0}=(\delta_{0},\theta_{0},\mu_{0},\nu_{0},\sigma_{0})\in\mathcal{Z}$,
\begin{equation*}
	\mathbb{E}\left[\int_{t}^{T}\big|\pi_{s}\big|^{2}+\big|c_{s}\big|^{2}ds\bigg|\xi=\xi_{0}\right]<\infty.
\end{equation*}
For brevity, we denote $\mathcal{A}^{MF}_{t}$ as the set of all admissible controls over $[t,T]$.
\end{Definition}
\begin{Definition}[DF strategy]\label{DF829}
	
	A pair $\left(\Pi^{\xi},C^{\xi}\right)$ is said to be a DF strategy, if  $\Pi^{\xi}(t)$ and $C^{\xi}(t,x,\overline{x})$ are of the following forms:
	\begin{equation*}
		\begin{split}
			&\Pi^{\xi}(t)=\sum_{k=1}^{N}\Pi^{k}_{1}(\xi)\Pi^{k}_{2}(t),\quad (t,\xi)\in[0,T]\times\mathcal{Z}, \\ &C^{\xi}\left(t,x,\overline{x}\right)=p_{1}(t)x+p_{2}(t,\xi)\overline{x}+q(t,\xi),\quad (t,x,\overline{x},\xi)\in[0,T]\times\mathbb{R}^{2}\times\mathcal{Z},
		\end{split}
	\end{equation*}
	where $N$ is a positive integer, $\Pi^{k}_{1}:\mathcal{Z}\rightarrow\mathbb{R}$,  $p_{2}$ and $q:[0,T]\times\mathcal{Z}\rightarrow\mathbb{R}$ are Borel measurable functions, and $\Pi^{k}_{2}$, $p_{1}:[0,T]\rightarrow\mathbb{R}$ are continuous functions. Moreover, for any  $\xi_{0}=(\delta_{0},\theta_{0},\mu_{0},\nu_{0},\sigma_{0})\in\mathcal{Z}$,   $p_{2}(\cdot,\xi_{0})$ and $q(\cdot,\xi_{0}):[0,T]\rightarrow\mathbb{R}$ are continuous functions; $\mathbb{E}\left[p_{2}(\cdot,\xi)\right]$ and $\mathbb{E}\left[q(\cdot,\xi)\right]:[0,T]\rightarrow\mathbb{R}$ are also continuous.  We denote the set of DF strategies by $\mathcal{S}_{MF}$. Similarly, a DF strategy $\left(\Pi^{\xi},C^{\xi}\right)\in\mathcal{S}_{MF}$ is said to be simple if  $C^{\xi}$ does not depend on $\overline{x}$.
\end{Definition}
\begin{remark}
	We explain that the special structure of the DF strategy is for the convenience of deriving the dynamics of the average wealth process $\overline{X}$.  We note that one can not simply take conditional expectation over the differential form of a SDE.
\end{remark}

Now we formulate the representative agent's optimization problem. Note that this is a mean field game with common noise $B$, so conditional  expectations given $B$ will be involved. As argued in \cite{2014Mean, 2020Many}, conditionally on the Brownian motion B, we can get some kind of  law of large numbers and asymptotic independence between the agents as $n\rightarrow\infty$, which suggests that the average wealth $\overline{X}_{t}$ and consumption $\overline{c}_{t}$ should be  $\mathbb{F}^{B}$-adapted processes. Then, the expected payoff of the representative agent is
\begin{equation}
\mathbb{E}\left[\int_{t_{0}}^{T}\lambda(t-t_{0})U(c_{t},\overline{c}_{t})dt+\lambda(T-t_{0})U(X^{\xi}_{T},\overline{X}_{T})\right],
\end{equation}
where $\lambda(\cdot)$ is a discount function defined in  \cref{disc} and $U(x,m):=-\exp\{-\frac{1}{\delta}(x-\theta m)\}$. The objective of the
 representative agent is to find the open-loop consistent control.
\begin{Definition} [Integrability condition]\label{Assumption3.2}
	
Let $(t_{0},x_{0})\in[0,T)\times\mathbb{R}$ be given.	We say that an admissible control $(\pi,c)\in\mathcal{A}^{MF}_{t_{0}}$ and average processes $(\overline{X},\overline{c})$ satisfy the integrability condition over $[t_{0},T]$, if for any deterministic sample $\xi_{0}=(\delta_{0},\theta_{0},\mu_{0},\nu_{0},\sigma_{0})\in\mathcal{Z}$, we have
	\begin{equation*}
		\mathbb{E}\left[\int_{t_{0}}^{T}\big|U(c_{t},\overline{c}_{t})\big|^{2}dt+\big|U(X^{\xi}_{T},\overline{X}_{T})\big|^{2}\bigg|\xi=\xi_{0}\right]<\infty,
	\end{equation*}
where $X^{\xi}$ is defined by \cref{mse} with the initial condition $(t_{0},x_{0})$ and the admissible control $(\pi,c)$.
\end{Definition}

As the equilibrium $(\pi^{*,\xi},c^{*,\xi})$  should lead to $\mathbb{E}\left[X^{ *,\xi}_{t}\big|\mathcal{F}^{B}_{t}\right]=\overline{X}_{t}$ and $\mathbb{E}\left[c^{*,\xi}_{t}\big|\mathcal{F}^{B}_{t}\right]=\overline{c}_{t}$, we formalize this discussion in the following definition.

\begin{Definition}[Mean Field Equilibrium]\label{mfe}
	
	Let $(\pi^{*,\xi},c^{*,\xi})$ be an admissible control over $[t_{0},T]$, and consider the $\mathbb{F}^{B}$-adapted processes $\overline{X}_{t}=\mathbb{E}\big[X^{*,\xi}_{t}\big|\mathcal{F}^{B}_{t}\big]$ and $\overline{c}_{t}=\mathbb{E}\big[c^{*,\xi}_{t}\big|\mathcal{F}^{B}_{t}\big]$, where $X^{*,\xi}$ is the wealth process corresponding to the control $(\pi^{*,\xi},c^{*,\xi})$ with initial condition $(t_{0},x_{0})$. We say that $(\pi^{*,\xi},c^{*,\xi})$ is a mean field equilibrium with respect to the initial condition $(t_{0},x_{0})$ if $(\pi^{*,\xi_{0}},c^{*,\xi_{0}})$ is an open-loop consistent control corresponding to this choice of $\overline{X}$ and $\overline{c}$ for any deterministic sample $\xi_{0}=(\delta_{0},\theta_{0},\mu_{0},\nu_{0},\sigma_{0})\in\mathcal{Z}$ with respect to the initial condition $(t_{0},x_{0})$.
\end{Definition}
\begin{Definition}[DF equilibrium strategy]
	
		A DF strategy $\left(\Pi^{*,\xi},C^{*,\xi}\right)\in\mathcal{S}_{MF}$ is said to be a DF equilibrium strategy if for every $(t_{0},x_{0})\in[0,T)\times\mathbb{R}$, the corresponding outcome $\left(\pi^{*,\xi}, c^{*,\xi}\right)\in\mathcal{A}^{MF}_{t_{0}}$ is a mean field equilibrium with respect to the initial condition $(t_{0},x_{0})$. Moreover,  a DF equilibrium strategy $\left(\Pi^{*,\xi},C^{*,\xi}\right)$ is said to be simple if the DF strategy $\left(\Pi^{*,\xi},C^{*,\xi}\right)$ is simple.
\end{Definition}
\begin{theorem}\label{3.2}
	Assume that, a.s., $\delta>0$, $\theta\in[0,1)$, $\mu>0$, $\sigma\geq0$, $\nu\geq0$, and $\sigma+\nu>0$. Define the constants
	    \begin{equation*}
	   		\psi:=\mathbb{E}\left[\theta\frac{\sigma^{2}}{\sigma^{2}+\nu^{2}}\right]\quad\text{and}\quad
	   		\phi:=\mathbb{E}\left[\delta\frac{\mu\sigma}{\sigma^{2}+\nu^{2}}\right],
	   \end{equation*}
	 where we assume that both expectations are finite.
	 Then there exists a unique simple DF equilibrium strategy $(\Pi^{*,\xi},C^{*,\xi})$ taking the following form:
	 \begin{align}
	 	&\Pi^{*,\xi}(t)=\left[\delta\frac{\mu}{\sigma^{2}+\nu^{2}}+\theta\frac{\sigma}{\sigma^{2}+\nu^{2}}\frac{\phi}{1-\psi}\right](T+1-t),\label{3.7}\\
 	    &C^{*,\xi}(t,x)=\frac{x}{T+1-t}-\delta H^{\xi}(t)-\theta\frac{\mathbb{E}\left[\delta H^{\xi}(t)\right]}{1-\mathbb{E}\left[\theta\right]}-\left[\delta+\theta\frac{\mathbb{E}[\delta ]}{1-\mathbb{E}[\theta]}\right]\ln\left[\lambda(T-t)\right],\label{3.8}
	 \end{align}
 where $H^{\xi}(t)$ is the randomization of $H^{i}(t)$, given by 
  \begin{equation}
 	H^{\xi}(t):=\frac{D}{2}\left[\frac{1}{T+1-t}-(T+1-t)\right]-\frac{1}{T+1-t}\int_{t}^{T}\ln\left[\lambda(T-s)\right]ds,
 \end{equation}
 where
 \begin{equation}
 	\begin{split}
 D:&=\frac{1}{2}\frac{\left(\mu+\sigma A\right)^{2}}{\nu^{2}+\sigma^{2}}-\frac{1}{2}\left(A\right)^{2}-B,\\
 A:&=\frac{\theta}{\delta}\mathbb{E}\left[\delta\frac{\sigma\mu}{\sigma^{2}+\nu^{2}}+\theta\frac{\sigma^{2}}{\sigma^{2}+\nu^{2}}\frac{\phi}{1-\psi}\right],\\
 B:&=\frac{\theta}{\delta}\mathbb{E}\left[\delta\frac{\mu^{2}}{\sigma^{2}+\nu^{2}}+\theta\frac{\sigma\mu}{\sigma^{2}+\nu^{2}}\frac{\phi}{1-\psi}\right],\\
\end{split}
\end{equation}
and we assume that all the expectations are finite. 
\end{theorem}

Similarly, we highlight the single stock case, noting that the form of the solution is essentially the same as in the $n$-agent games, presented in \cref{2.2c}.
\begin{Corollary}[Single stock]\label{3.2c}
	Suppose that $(\mu,\nu,\sigma)$ are deterministic, with $\nu=0$ and $\mu,\sigma>0$.
	Then the strategy $(\Pi^{*,\xi},C^{*,\xi})$ has the following form:
	\begin{align}
		&\Pi^{*,\xi}(t)=\frac{\mu}{\sigma^{2}}\left(\delta+\theta\frac{\mathbb{E}[\delta]}{1-\mathbb{E}[\theta]}\right)\left(T+1-t\right),\\
		&C^{*,\xi}(t,x)=\frac{x}{T+1-t}+\left(\delta+\theta\frac{\mathbb{E}[\delta]}{1-\mathbb{E}[\theta]}\right)\mathcal{H}(t)-\left(\delta+\theta\frac{\mathbb{E}[\delta]}{1-\mathbb{E}[\theta]}\right)\ln\left[\lambda(T-t)\right],	
	\end{align} 
where the function $\mathcal{H}(\cdot)$ is the same as \cref{fH}.
\end{Corollary}		
\begin{proof}[Proof of \cref{3.2}]
	
First, observe that it suffices to restrict our attention to stochastic processes $\overline{X}$ and $\overline{c}$ of the form $\overline{X}_{t}=\mathbb{E}\big[X^{\xi}_{t}\big|\mathcal{F}^{B}_{t}\big]$ and $\overline{c}_{t}=\mathbb{E}\big[c_{t}\big|\mathcal{F}^{B}_{t}\big]$, where $X^{\xi}$ is defined by \cref{mse} with the initial condition $(t_{0},x_{0})$ and the admissible control $(\pi,c)\in\mathcal{A}^{MF}_{t_{0}}$. Moreover, we assume that $\left(\pi,c\right)$ is the outcome of a DF strategy $\left(\Pi^{\xi},C^{\xi}\right)\in\mathcal{S}_{MF}$:
\begin{equation*}
	\pi_{t}=\Pi^{\xi}(t)=\sum_{k=1}^{N}\Pi^{k}_{1}(\xi)\Pi^{k}_{2}(t),\,\ c_{t}=C^{\xi}(t,X^{\xi}_{t},\overline{X}_{t})=p_{1}(t)X^{\xi}_{t}+p_{2}(t,\xi)\overline{X}_{t}+q(t,\xi),\,t\in[t_{0},T],
\end{equation*}
where $\Pi^{k}_{1}$, $\Pi^{k}_{2}$, $p_{1}$, $p_{2}$ and $q$ are defined in \cref{DF829}.

As $\xi=(\delta,\theta,\mu,\nu,\sigma)$, $W$ and $B$ are independent, 
 we have $$\overline{c}_{t}=\overline{C}(t,\overline{X}_{t}):=\mathbb{E}\big[C^{\xi}(t,X^{\xi}_{t},\overline{X}_{t})\big|\mathcal{F}^{B}_{t}\big]=\left(p_{1}(t)+\mathbb{E}\left[p_{2}(t,\xi)\right]\right)\overline{X}_{t}+\mathbb{E}\left[q(t,\xi)\right]$$ and the dynamic of $\overline{X}$ is \footnote{When $(\pi,c)$ is the outcome of a DF strategy, the dynamic of $X^{\xi}$ is linear. Hence, it is standard to derive the dynamics of $\overline{X}$.} 
\begin{equation}\label{MFEmeanwealth}
	\begin{cases}
	d\overline{X}_{t}=\left[-\left(p_{1}(t)+\mathbb{E}\left[p_{2}(t,\xi)\right]\right)\overline{X}_{t}+\mathbb{E}\left[\Pi^{\xi}(t)\mu\right]-\mathbb{E}\left[q(t,\xi)\right]\right]dt\\\quad\quad+E\left[\Pi^{\xi}(t)\sigma\right]dB_{t}, \quad t\in[t_{0},T], \\
	\overline{X}_{t_{0}}=x_{0}.
	\end{cases}
\end{equation}
We now find the open-loop consistent control for the representative agent, then we  resolve the resulting fixed point problem to obtain the desired equilibrium.

\textit{Open-loop consistent control for the representative agent}.  Let $\xi_{0}=(\delta_{0},\theta_{0},\mu_{0},\nu_{0},\sigma_{0})$ represent a deterministic sample from its random type distribution for now.
According to  \cref{verifmean} in \cref{Appendix}, we aim to find a classical solution of the following  PDE:
\begin{equation}
	\begin{cases}
		V^{\xi_{0}}_{t}(t,x,\overline{x})+\hat{\Pi}^{\xi_{0}}(t,x,\overline{x})\left(V^{\xi_{0}}_{x}(t,x,\overline{x})\mu_{0}+V^{\xi_{0}}_{x\overline{x}}(t,x,\overline{x})\sigma_{0}\mathbb{E}\left[\Pi^{\xi}(t)\sigma\right]\right)\\+\frac{1}{2}\hat{\Pi}^{\xi_{0}}(t,x,\overline{x})^{2}V^{\xi_{0}}_{xx}(t,x,\overline{x})\left(\nu_{0}^{2}+\sigma_{0}^{2}\right)-\hat{C}^{\xi_{0}}(t,x,\overline{x})V^{\xi_{0}}_{x}(t,x,\overline{x})\\+V^{\xi_{0}}_{\overline{x}}(t,x,\overline{x})\left(\mathbb{E}\left[\Pi^{\xi}(t)\mu\right]-\overline{C}(t,\overline{x})\right)\\+\frac{1}{2}V^{\xi_{0}}_{\overline{x}\overline{x}}(t,x,\overline{x})\left(\mathbb{E}\left[\Pi^{\xi}(t)\sigma\right]\right)^{2}=0,\quad (t,x,\overline{x})\in[0,T]\times\mathbb{R}^{2},\\
		V^{\xi_{0}}(T,x,\overline{x})=\frac{1}{\delta_{0}}\exp\left(-\frac{1}{\delta_{0}}x+\frac{\theta_{0}}{\delta_{0}}\overline{x}\right),\quad (x,\overline{x})\in\mathbb{R}^{2},\\
		\hat{\Pi}^{\xi_{0}}(t,x,\overline{x})=-\frac{\mu_{0} V^{\xi_{0}}(t,x,\overline{x})+\sigma_{0}\mathbb{E}\left[\Pi^{\xi}(t)\sigma\right]V^{\xi_{0}}_{\overline{x}}(t,x,\overline{x})}{\left(\nu_{0}^{2}+\sigma_{0}^{2}\right)V^{\xi_{0}}_{x}(t,x,\overline{x})},\quad (t,x,\overline{x})\in[0,T]\times\mathbb{R}^{2},\\
		\hat{C}^{\xi_{0}}(t,x,\overline{x})=-\delta_{0}\ln\left[\delta_{0}\lambda(T-t)V^{\xi_{0}}(t,x,\overline{x})\right]+\theta_{0}\overline{C}(t,\overline{x}),\quad (t,x,\overline{x})\in[0,T]\times\mathbb{R}^{2}.
	\end{cases}\label{sys2mean}
\end{equation}

We consider the following ansatz based on the terminal condition:
\begin{equation}
	V^{\xi_{0}}(t,x,\overline{x})=\frac{1}{\delta_{0}}\exp\left\{f^{\xi_{0}}(t)x+g^{\xi_{0}}(t)\overline{x}+h^{\xi_{0}}(t)\right\},\,       (t,x,\overline{x})\in[0,T]\times\mathbb{R}^{2},\label{Guassmean}
\end{equation} 
where $f^{\xi_{0}}(\cdot),g^{\xi_{0}}(\cdot)$ and $h^{\xi_{0}}(\cdot)\in\ C^{1}\left([0,T];\mathbb{R}\right)$ such that $f^{\xi_{0}}(T)=-\frac{1}{\delta_{0}}$, $g^{\xi_{0}}(T)=\frac{\theta_{0}}{\delta_{0}}$ and $h^{\xi_{0}}(T)=0$.

Putting \cref{Guassmean} into \cref{sys2mean}, we have
\begin{equation*}
	\begin{split}
	&\left[\dot{f}^{\xi_{0}(t)}+\delta_{0}f^{\xi_{0}}(t)^{2}\right]x+\left[\dot{g}^{\xi_{0}}(t)+\delta_{0} f^{\xi_{0}}(t)g^{\xi_{0}}(t)\right]\overline{x}\\
	&+\dot{h}^{\xi_{0}}(t)+\delta_{0} f^{\xi_{0}}(t)h^{\xi_{0}}(t)+\delta_{0} f^{\xi_{0}}(t)\ln\left[\lambda(T-t)\right]\\&-\frac{1}{2}\frac{\left(\mu_{0}+\sigma_{0} g^{\xi_{0}}(t)\mathbb{E}\left[\sigma\Pi^{\xi}(t)\right]\right)^{2}}{\nu_{0}^{2}+\sigma_{0}^{2}}+g^{\xi_{0}}(t)\mathbb{E}\left[\mu\Pi^{\xi}(t)\right]+\frac{1}{2}\left(g^{\xi_{0}}(t)\right)^{2}\left(\mathbb{E}\left[\sigma\Pi^{\xi}(t)\right]\right)^{2}\\&-\left(\theta_{0} f^{\xi_{0}}(t)+g^{\xi_{0}}(t)\right)\overline{C}(t,\overline{x})=0.
	\end{split}
\end{equation*}
If $f^{\xi_{0}}(\cdot)$ and $g^{\xi_{0}}(\cdot)$ solve the following ODEs:
\begin{equation}
	\begin{cases}
		\dot{f}^{\xi_{0}}(t)+\delta_{0}f^{\xi_{0}}(t)^{2}=0,\quad t\in[0,T],\\
		\dot{g}^{\xi_{0}}(t)+\delta_{0}f^{\xi_{0}}(t)g^{\xi_{0}}(t)=0,\quad t\in[0,T],\\
		f^{\xi_{0}}(T)=-\frac{1}{\delta_{0}},\, g^{\xi_{0}}(T)=\frac{\theta_{0}}{\delta_{0}},
	\end{cases}
\end{equation}
which is explicitly solved by 
\begin{equation}\label{3.15}
	\begin{cases}
		f^{\xi_{0}}(t)=-\frac{1}{\delta_{0}}\frac{1}{T+1-t},\\
		g^{\xi_{0}}(t)=\frac{\theta_{0}}{\delta_{0}}\frac{1}{T+1-t},\\
	\end{cases}
\end{equation}
then it holds that $\theta_{0} f^{\xi_{0}}(t)+g^{\xi_{0}}(t)=0$ for any $t\in[0,T]$. Hence, we obtain the ODE for $h^{\xi_{0}}(t)$:
\begin{equation*}
\begin{split}
	&\dot{h}^{\xi_{0}}(t)+\delta_{0} f^{\xi_{0}}(t)h^{\xi_{0}}(t)+\delta_{0} f^{\xi_{0}}(t)\ln\left[\lambda(T-t)\right]\\&-\frac{1}{2}\frac{\left(\mu_{0}+\sigma_{0} g^{\xi_{0}}(t)\mathbb{E}\left[\sigma\Pi^{\xi}(t)\right]\right)^{2}}{\nu_{0}^{2}+\sigma_{0}^{2}}+g^{\xi_{0}}(t)\mathbb{E}\left[\mu\Pi^{\xi}(t)\right]+\frac{1}{2}\left(g^{\xi_{0}}(t)\right)^{2}\left(\mathbb{E}\left[\sigma\Pi^{\xi}(t)\right]\right)^{2}=0
\end{split}
\end{equation*}
with the terminal condition $h^{\xi_{0}}(T)=0$, whose solution is given by 
\begin{equation}
		h^{\xi_{0}}(t)=\frac{1}{T+1-t}\int_{t}^{T}(T+1-s)G^{\xi_{0}}(\Pi^{\xi}(s),s)ds,
\end{equation}
where
\begin{equation}
	\begin{split}
 G^{\xi_{0}}(\pi,t):=-\frac{1}{T+1-t}\ln\left[\lambda(T-t)\right]-\frac{1}{2}\frac{\left(\mu_{0}+\frac{\theta_{0}\sigma_{0}}{(T+1-t)\delta_{0}}\mathbb{E}[\sigma\pi]\right)^{2}}{\nu_{0}^{2}+\sigma_{0}^{2}}\\+\frac{\theta_{0}}{(T+1-t)\delta_{0}}\mathbb{E}[\mu\pi]+\frac{\theta_{0}^{2}}{2(T+1-t)^{2}\delta^{2}_{0}}(\mathbb{E}[\sigma\pi])^{2}.
 	\end{split}
\end{equation}

Then, by the representation of $(\hat{\Pi}^{\xi_{0}},\hat{C}^{\xi_{0}})$ in \cref{sys2mean}, we have
\begin{equation}\label{meanweakop}
	\begin{split}
	\hat{\Pi}^{\xi_{0}}(t,x,\overline{x})&=\frac{\mu_{0}+\sigma_{0}\frac{\theta_{0}}{\delta_{0}}\frac{1}{T+1-t}\mathbb{E}[\sigma\Pi^{\xi}(t)]}{\left(\nu_{0}^{2}+\sigma_{0}^{2}\right)\left(\frac{1}{\delta_{0}}\frac{1}{T+1-t}\right)}=\frac{\delta_{0}\mu_{0}}{\nu_{0}^{2}+\sigma_{0}^{2}}\left(T+1-t\right)+\frac{\sigma_{0}\theta_{0}}{\nu_{0}^{2}+\sigma_{0}^{2}}\mathbb{E}[\sigma\Pi^{\xi}(t)],\\
	\hat{C}^{\xi_{0}}(t,x,\overline{x})&=\frac{1}{T+1-t}x-\theta_{0}\frac{1}{T+1-t}\overline{x}-\delta_{0} h^{\xi_{0}}(t)-\delta_{0}\ln\left[\lambda(T-t)\right]+\theta_{0} \overline{C}(t,\overline{x}).
	\end{split}
\end{equation}

Note that $\hat{\Pi}^{\xi_{0}}(t,x,\overline{x})=\hat{\Pi}^{\xi_{0}}(t)$ is independent of the state argument $x$ and $\overline{x}$.  In the same manner as the $n$-agent case, we obtain that for $\xi_{0}$-type agent, the outcome $\left(\hat{\pi}^{\xi_{0}},\hat{c}^{\xi_{0}}\right)$ associated with $\left(\hat{\Pi}^{\xi_{0}},\hat{C}^{\xi_{0}}\right)$ is an open-loop consistent control with respect to the initial condition $(t_{0},x_{0})$.

Thus, the open-loop consistent control $\left(\hat{\pi}^{\xi},\hat{c}^{\xi}\right)$ for the representative agent is 
\begin{align}
	&\hat{\pi}^{\xi}_{t}=\hat{\Pi}^{\xi}(t)=\frac{\delta\mu}{\nu^{2}+\sigma^{2}}\left(T+1-t\right)+\frac{\sigma\theta}{\nu^{2}+\sigma^{2}}\mathbb{E}[\sigma\Pi^{\xi}(t)],\quad t\in[t_{0},T],\label{mfepi2}\\
	&\hat{c}^{\xi}_{t}=\hat{C}^{\xi}(t,\hat{X}^{\xi}_{t},\overline{X}_{t})=\frac{1}{T+1-t}\hat{X}^{\xi}_{t}-\theta\frac{1}{T+1-t}\overline{X}_{t}-\delta h^{\xi}(t)\label{mfec2}-\delta\ln\left[\lambda(T-t)\right]\\&+ \theta \overline{C}(t,\overline{X}_{t}),\quad t\in[t_{0},T],\nonumber
\end{align}
where $\left(\hat{X}^{\xi}, \overline{X}\right)$ is defined by
\begin{equation*}
	\begin{cases}
		d\hat{X}^{\xi}_{t}=\left(\hat{\Pi}^{\xi}(t)\mu-\hat{C}^{\xi}(t,\hat{X}^{\xi}_{t},\overline{X}_{t})\right)dt+\hat{\Pi}^{\xi}(t)\nu dW_{t}+\hat{\Pi}^{\xi}(t)\sigma dB_{t},\\
		d\overline{X}_{t}=\left[-\left(p_{1}(t)+\mathbb{E}\left[p_{2}(t,\xi)\right]\right)\overline{X}_{t}+\mathbb{E}\left[\Pi^{\xi}(t)\mu\right]-\mathbb{E}\left[q(t,\xi)\right]\right]dt+E\left[\Pi^{\xi}(t)\sigma\right]dB_{t},\\
		\hat{X}^{\xi}_{t_{0}}=x_{0},\,\overline{X}_{t_{0}}=x_{0}.
	\end{cases}
\end{equation*}

Moreover, if we assume that all the expectations associated with the random type vector $\xi$ is finite, then we get that  $\left(\hat{\Pi}^{\xi},\hat{C}^{\xi}\right)$  is also a DF strategy. Hence, we have constructed a best response map 
\begin{equation}
	F:\left(\Pi^{\xi},C^{\xi}\right)\rightarrow\left(\hat{\Pi}^{\xi},\hat{C}^{\xi}\right).
\end{equation}

\textit{Fixed point problem}. We first address the investment strategies. For the candidate investment strategy $\Pi^{\xi}$ to be a fixed point, we need $\Pi^{\xi}(t)=\hat{\Pi}^{\xi}(t)$, for $t\in[0,T]$. In light of \cref{mfepi2}, we have
\begin{equation}\label{3.18}
	\hat{\Pi}^{\xi}(t)=\delta\frac{\mu}{\nu^{2}+\sigma^{2}}(T+1-t)+\theta\frac{\sigma}{\nu^{2}+\sigma^{2}}\mathbb{E}[\sigma\hat{\Pi}^{\xi}(t)].
\end{equation}

Multiply both sides of equation \cref{3.18} by $\sigma$ and average to find that $\mathbb{E}[\sigma\hat{\Pi}^{\xi}(t)]$ must satisfy the following fixed point equation:
\begin{equation}
	\mathbb{E}[\sigma\hat{\Pi}^{\xi}(t)]=\phi(T+1-t)+\psi\mathbb{E}[\sigma\hat{\Pi}^{\xi}(t)].\label{mfespi}
\end{equation}
We then have the following cases to get the fixed point:
\begin{itemize}
\item[(i)] If $\psi<1$, then \cref{mfespi}  yields $\mathbb{E}[\sigma\hat{\Pi}^{\xi}(t)]=\left[\phi/\left(1-\psi\right)\right]\left(T+1-t\right)$, and the investment strategy is given by \cref{3.7}.

\item[(ii)] If $\psi=1$, then the equation \cref{mfespi} has no solution. Note that $\psi=1$ and $\phi=0$ cannot happen. By assumption $\delta>0$, $\mu>0$, and $\sigma+\nu>0$, one can easily get a contradiction.
\end{itemize}
Next, we address the consumption strategies. Similarly, the candidate consumption strategy  $\hat{C}^{\xi}$ need to satisfy that 
\begin{equation*}
	\begin{split}
			\hat{C}^{\xi}(t,\hat{X}^{\xi}_{t},\overline{X}_{t})&=\frac{1}{T+1-t}\hat{X}^{\xi}_{t}-\theta\frac{1}{T+1-t}\overline{X}_{t}-\delta\hat{h}^{\xi}(t)\\&+\theta\mathbb{E}\left[	\hat{C}^{\xi}(t,\hat{X}^{\xi}_{t},\overline{X}_{t})\big|\mathcal{F}^{B}_{t}\right]-\delta\ln\left[\lambda(T-t)\right],
	\end{split}
\end{equation*}
where $\hat{h}^{\xi}(t)=\frac{1}{T+1-t}\int_{t}^{T}(T+1-s)G^{\xi}(\Pi^{*,\xi}(s),s)ds$. In fact, using the result in \cref{h}, we have $\hat{h}^{\xi}(t)=H^{\xi}(t)$. To avoid confusion, we use $H^{\xi}(t)$ instead of $\hat{h}^{\xi}(t)$. Then we should solve the following fixed point problem:
\begin{equation}\label{3.20}
	\begin{split}
	\hat{C}^{\xi}(t,\hat{X}^{\xi}_{t},\overline{X}_{t})&=\frac{1}{T+1-t}\hat{X}^{\xi}_{t}-\theta\frac{1}{T+1-t}\overline{X}_{t}-\delta H^{\xi}(t)\\&+\theta\mathbb{E}\left[	\hat{C}^{\xi}(t,\hat{X}^{\xi}_{t},\overline{X}_{t})\big|\mathcal{F}^{B}_{t}\right]-\delta\ln\left[\lambda(T-t)\right].
	\end{split}
\end{equation}

Taking the conditional expectation given $\mathcal{F}^{B}_{t}$ both sides of \cref{3.20}  to find that 
\begin{equation}
	\begin{split}
		\mathbb{E}\left[\hat{C}^{\xi}(t,\hat{X}^{\xi}_{t},\overline{X}_{t})\big|\mathcal{F}^{B}_{t}\right]&=\frac{1-\mathbb{E}[\theta]}{T+1-t}\overline{X}_{t}-\mathbb{E}[\delta H^{\xi}(t)]+\mathbb{E}[\theta]	\mathbb{E}\left[\hat{C}^{\xi}(t,\hat{X}^{\xi}_{t},\overline{X}_{t})\big|\mathcal{F}^{B}_{t}\right]\\&-\mathbb{E}[\delta]\ln\left[\lambda(T-t)\right],\label{meqc2}
	\end{split}
\end{equation}
then we have the following cases to find the fixed point:
\begin{itemize}
\item[(i)] If $\mathbb{E}[\theta]<1$, then \cref{meqc2} yields $\mathbb{E}\left[\hat{C}^{\xi}(t,\hat{X}^{\xi}_{t},\overline{X}_{t})\big|\mathcal{F}^{B}_{t}\right]=\frac{\overline{X}_{t}}{T+1-t}-\frac{\mathbb{E}[\delta H^{\xi}(t)]}{1-\mathbb{E}[\theta]}-\frac{\mathbb{E}[\delta]\ln\left[\lambda(T-t)\right]}{1-\mathbb{E}[\theta]}$, and the consumption stratgey is given by \cref{3.8}.

\item[(ii)] If $\mathbb{E}[\theta]=1$ and $\mathbb{E}[\delta H(t)]+\mathbb{E}[\delta]	\ln\left[\lambda(T-t)\right]\not\equiv0$, then equation \cref{meqc2} has no solution.

\item[(iii)] If $\mathbb{E}[\theta]=1$ and $\mathbb{E}[\delta H(t)]+\mathbb{E}[\delta]	\ln\left[\lambda(T-t)\right]\equiv0$, then there exist infinitely many solutions.
\end{itemize}

 We limit $\theta$ in $[0,1)$ to remove `bad' cases. In summary, we get that there exists a unique solution to the fixed point problem, which turns out to be a simple DF equilibrium strategy $(\Pi^{*,\xi}, C^{*,\xi})$ given by \cref{3.7} and \cref{3.8}.  Similar to the discussion on uniqueness in the proof of \cref{The2.1}, we can also conclude that $(\Pi^{*,\xi}, C^{*,\xi})$ is the unique simple DF equilibrium strategy.
\end{proof}
\section{Discussion of the equilibrium}\label{discussion}
 
We now discuss the interpretation of equilibria. We limit the discussion to the mean field case, for which the DF equilibrium strategy is given by  \cref{3.2} and  \cref{3.2c}, as $n$-agent equilibria have essentially the same structure.

First, the investment strategy $\Pi^{*,\xi}$ is $\mathcal{F}^{MF}_{0}$-measurable and wealth independent, meaning that as the agent gets richer, she will decrease the proportion of investment in the risky asset. It's worth noting that the investment strategy  $\Pi^{*,\xi}$ is independent of the discount function, which is consistent with the result of \cite{EP14900367320210201} and \cite{MS-2010}. Moreover, $\Pi^{*,\xi}$ consists of two components. The first,
$\frac{\delta\mu}{\sigma^{2}+\nu^{2}}\left(T+1-t\right)$, is the classical Merton portfolio. The second component is always nonnegative, vanishing only when $\theta=0$, which means no competition. It is clear that with increasing the competition weight $\theta$, the agent will increase the allocation in the risky asset. Note that $\Pi^{*,\xi}$ is a linear function of $(T+1-t)$, and  the coefficient is equal to the solution of the MFG in \cite{D2017Mean}, where the consumption is not considered. Then it is obvious to see that as time goes on, the agent will invest less in the risky asset, and at  terminal time $T$, the amount invested in the risky asset drops down to the equilibrium portfolio amount in \cite{D2017Mean}. For more analysis of the influence of the parameters, we refer the reader to  \cite{D2017Mean}.

We further restrict our attention to the single stock case of \cref{3.2c}, where the effects of the parameter are more transparent. Note that if $\theta=0$, then we recover the open-loop equilibrium strategy without competition, with $\Pi^{*,\xi}(t)=\frac{\mu}{\sigma^{2}}\delta(T+1-t)$ and $C^{*,\xi}(t,x)=\frac{x}{T+1-t}+\delta\mathcal{H}(t)-\delta\ln\left[\lambda(T-t)\right]$; see \cite{EP14900367320210201} for comparison. For the general $\theta$, we may still rewrite $\Pi^{*,\xi}$ and $C^{*,\xi}$ in an analogous manner as
\begin{align*}
	&\Pi^{*,\xi}(t)=\frac{\mu}{\sigma^{2}}\hat{\delta}\left(T+1-t\right),\\
	&C^{*,\xi}(t,x)=\frac{x}{T+1-t}+\hat{\delta}\mathcal{H}(t)-\hat{\delta}\ln\left[\lambda(T-t)\right],
\end{align*} 
where the \textit{effective risk tolerance parameter} is
\[\hat{\delta}=\delta+\theta\frac{\mathbb{E}[\delta]}{1-\mathbb{E}[\theta]}.\]
The parameter $\hat{\delta}$ has already appeared in \cite{D2017Mean} and \cite{2021Meanfieldito} and it is obvious to see that  $\hat{\delta}>\delta$ if $\theta>0$, the difference $\hat{\delta}-\delta$ increases with $\theta$, with $\mathbb{E}[\delta]$, and with $\mathbb{E}[\theta]$.

We consider the following hyperbolic discount function,
\[\lambda(t)=\left(1+\beta t\right)^{-\frac{\rho}{\beta}}, \,\, t\in[0,T],\]
where $\rho>0$, $\beta>0$.
In particular, $\lim_{\beta\rightarrow0}\lambda(t)=e^{-\rho t}$. In this case the DF equilibrium strategy reduces to the solution of Merton problem:
\begin{equation*}
\begin{split}
	&\Pi^{*,\xi}(t)=\frac{\mu}{\sigma^{2}}\hat{\delta}\left(T+1-t\right),\\
	&C^{*,\xi}(t,x)=\frac{x}{T+1-t}+\frac{\hat{\delta}}{2}\left[\frac{1}{2}\left(\frac{\mu}{\sigma}\right)^{2}+\rho\right]\left[(T+1-t)-\frac{1}{T+1-t}\right],
\end{split}
\end{equation*}
which is consistent with the fact that time-consistent equilibrium strategy under exponential discounting is nothing but the optimal strategy; see, e.g., \cite{EP14900367320210201, Bjork}.

Finally, we numerically compute the average consumption $\mathbb{E}[C^{*,\xi}(t,X^{*,\xi}_{t})]$. The parameters are $\mu=\sigma=1$, $t_{0}=0$, $x_{0}=10$, $\rho=0.1$ and $T=2$.  As we can see in \cref{p1},  the average consumption decreases with increasing $\beta$. Note that a larger $\beta$ means that the agent is more patient in the future; see, e.g., \cite{2019Non}. We can either say that the difference between the agent and her
future self is large when $\beta$ is large. In order to make an agreement, the sophisticated agent may give up more of her utility to the future. Hence, it is reasonable to spend less on consumption. In \cref{p2}, we investigate the impact of competitiveness and risk tolerance on the average consumption. It is clear that the average consumption increases with increasing $\mathbb{E}[\hat{\delta}]$. Since $\mathbb{E}[\hat{\delta}]$ increases when either $\mathbb{E}[\delta]$ increases or $\mathbb{E}[\theta]$ increases, we conclude that in an environment with high competition or high risk tolerance, the average consumption is also high.
\begin{figure}[tbhp]
	\centering
\includegraphics[scale=0.4]{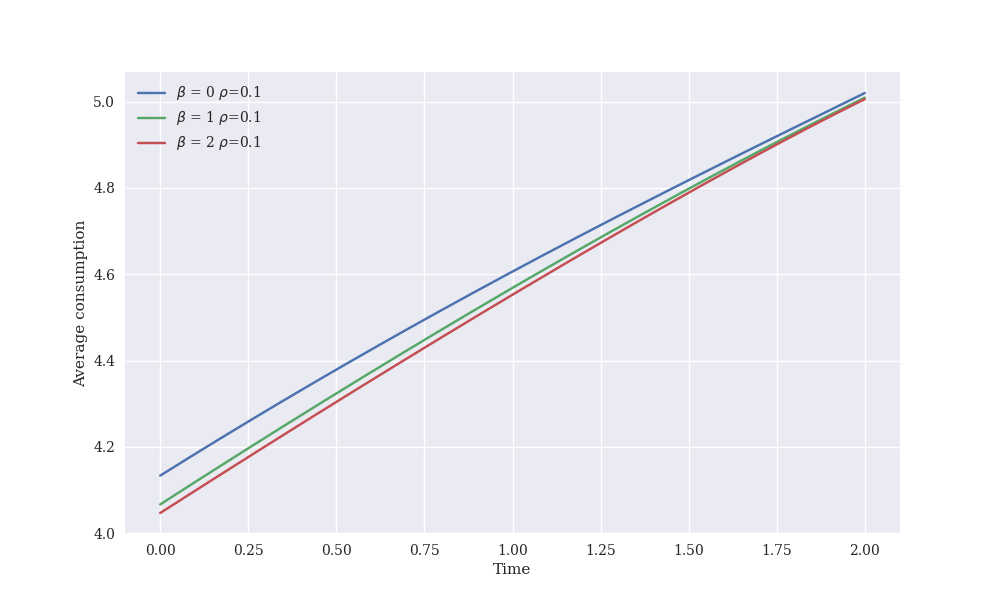}
\caption{Average consumption $\mathbb{E}[C^{*,\xi}(t,X^{*,\xi}_{t})]$ versus $t$ for various values of $\beta$.}
\label{p1}
\end{figure}
\begin{figure}[tbhp]
	\centering
	\includegraphics[scale=0.4]{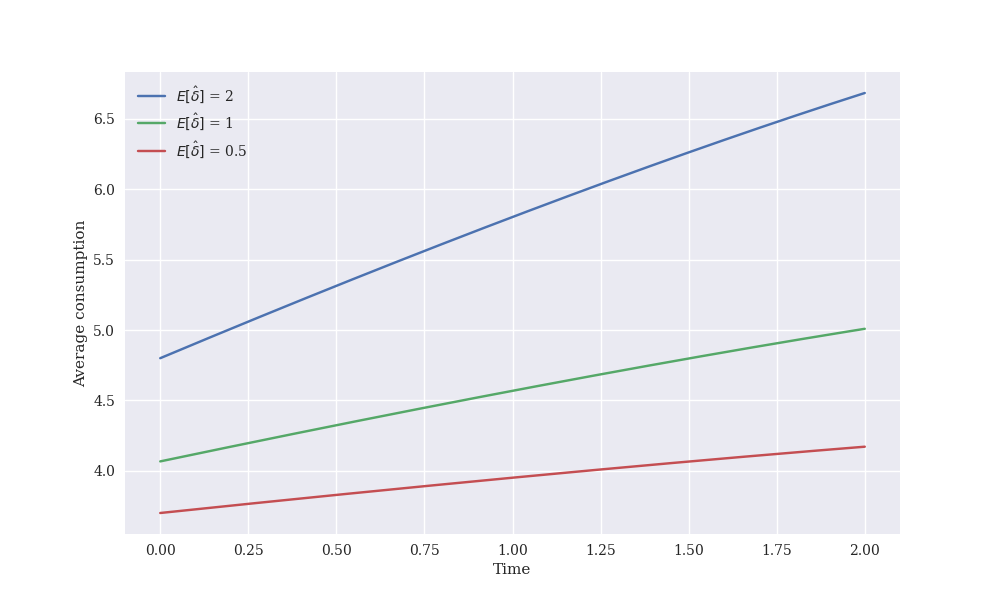}
	\caption{Average consumption $\mathbb{E}[C^{*,\xi}(t,X^{*,\xi}_{t})]$ versus $t$ for various values of $\mathbb{E}[\hat{\delta}]$.}
	\label{p2}
\end{figure}
\section{Conclusions  and future research}\label{sect5}

We have studied the MFGs and the $n$-agent games under CARA relative performance by allowing general discount functions. To deal with the time-inconsistency, each agent chooses the open-loop consistent strategy which is finally characterized by a PDE. By solving a fixed point problem,  we explicitly constructed the \textit{DF equilibrium strategy}. Next, we discuss some possible extensions for future research.

A natural extension is 
to consider the closed-loop equilibrium strategy, in which the equilibrium strategy will be characterized by a system of nonlocal ODEs  with complicated coupling structure. Both theoretical analysis and numerical solution will be a huge challenge.  It will be interesting to extend our current work to a model with general market parameters, e.g., the incomplete market model considered in \cite{2021Time}. A closed form solution may not exist. The existence of an MFE may require some different mathematical arguments.

\appendix

\section{The characterization of open-loop consistent control}\label{Appendix}\label{char}
\subsection{The $n$-agent games}
 In this section, we aim to characterize the open-loop consistent control by verification argument in the same spirit of \cite{EP14900367320210201}. Moreover, similar to classical time-consistent stochastic optimal control problem, we derive a PDE, which gives us a characterization of open-loop consistent control.

Let $i\in\left\{1,\cdots,n\right\}$, {$(t_{0},x_{0})\in[0,T)\times\mathbb{R}^{n}$} and $\left(\pi,c\right)^{(i)}=\left(\pi^{k},c^{k}\right)_{k\neq i}\in{\mathcal{A}^{n-1}_{t_{0}}}$ be given. We denote by $(\hat{\pi}^{i},\hat{c}^{i})\in{\mathcal{A}_{t_{0}}}$ a candidate control for agent $i$ and introduce the following BSDE defined on the interval $[{t_{0}},T]$:
\begin{equation}
	\begin{cases}
		dp(t)=q^{B}(t)dB_{t}+\sum_{k=1}^{n}q^{k}(t)dW^{k}_{t},\\
		p(T)=\frac{1}{\delta_{i}}\left(1-\frac{\theta_{i}}{n}\right)\exp\left\{-\frac{1}{\delta_{i}}\left(1-\frac{\theta_{i}}{n}\right)\hat{X}_{T}^{i}+\frac{\theta_{i}}{\delta_{i}}\overline{X}^{(i)}_{T}\right\},
	\end{cases}\label{BSDE1}
\end{equation}
where $\hat{X}^{i}$ and $\overline{X}^{(i)}$ are given by 
\begin{equation*}
	\begin{cases}
		d\hat{X}^{i}_{t}=\left(\hat{\pi}^{i}_{t}\mu_{i}-\hat{c}^{i}_{t}\right)dt+\hat{\pi}^{i}_{t}\nu_{i}dW_{t}^{i}+\hat{\pi}^{i}_{t}\sigma_{i}dB_{t},\,{t\in[t_{0},T]},\\
		\hat{X}^{i}_{{t_{0}}}=x^{i}_{0},
	\end{cases}
\end{equation*}
and
\begin{equation*}
	\begin{cases}
		d\overline{X}^{(i)}_{t}=(\overline{\mu\pi}^{(i)}_{t}-\overline{c}^{(i)}_{t})dt+\overline{\sigma\pi}^{(i)}_{t}dB_{t}+\frac{1}{n}\sum_{k\ne i}\nu_{k}\pi^{k}_{t}dW^{k}_{t}, \, {t\in[t_{0},T]},\\ \overline{X}^{(i)}_{{t_{0}}}=\overline{x}^{(i)}_{0}:=\frac{1}{n}\sum_{k\neq i}x_{0}^{k}.
	\end{cases}
\end{equation*}
If $U_{i}(\hat{X}_{T}^{i},\overline{X}^{(i)}_{T})=-\exp\left\{-\frac{1}{\delta_{i}}\left(1-\frac{\theta_{i}}{n}\right)\hat{X}_{T}^{i}+\frac{\theta_{i}}{\delta_{i}}\overline{X}^{(i)}_{T}\right\}\in L^{2}_{\mathcal{F}_{T}}(\Omega;\mathbb{R})$, then BSDE \cref{BSDE1} has a unique solution \[(p(\cdot),q(\cdot))\in\mathbb{S}^{2}_{\mathbb{F}}({t_{0}},T;\mathbb{R})\times\mathbb{H}^{2}_{\mathbb{F}}({t_{0}},T;\mathbb{R}^{n+1}),\]
where $q(\cdot)=(q^{B}(\cdot),q^{1}(\cdot),\dots,q^{n}(\cdot))$.

Then we have the following theorem:
\begin{theorem}\label{suff}
 Assume that  $U_{i}(\hat{X}_{T}^{i},\overline{X}^{(i)}_{T})\in L^{2}_{\mathcal{F}_{T}}(\Omega;\mathbb{R})$.
 Then, the $(\hat{\pi}^{i},\hat{c}^{i})$ is an open-loop consistent control for agent $i$ with respect to the {initial condition $(t_{0},x_{0})$} in response to $(\pi^{k},c^{k})_{k\neq i}$, if the following conditions hold 
\begin{align}
	&\mu_{i}p(t)+\nu_{i}q^{i}(t)+\sigma_{i}q^{B}(t)=0,\,\ t\in[{t_{0}},T],\label{5.2}\\
	&-\lambda(T-t)p(t)+\frac{1}{\delta_{i}}\left(1-\frac{\theta_{i}}{n}\right)\exp\left\{-\frac{1}{\delta_{i}}\left(1-\frac{\theta_{i}}{n}\right)\hat{c}_{t}^{i}+\frac{\theta_{i}}{\delta_{i}}\overline{c}^{(i)}_{t}\right\}=0,\,\ t\in[{t_{0}},T].\label{5.3}
\end{align}
\end{theorem}  
\begin{proof}[Proof of \cref{suff}]
	Suppose that $(\hat{\pi}^{i},\hat{c}^{i})$ satisfies the conditions \cref{5.2}-\cref{5.3}. For any {$t\in[t_{0},T]$} and $\epsilon\in(0,T-t)$, we consider $(\pi^{i,t,\epsilon},c^{i,t,\epsilon})$ by \cref{eps} and define $X^{i,t,\epsilon}$ as the agent $i$'s wealth process with the initial condition $(t_{0},x^{i}_{0})$ and the control $(\pi^{i,t,\epsilon},c^{i,t,\epsilon})$, then we have the following difference
	\begin{align*}
		\Delta(J_{i}&):=J_{i}(t,\hat{X}^{i}_{t},\overline{X}^{(i)}_{t},(\hat{\pi}^{i},\hat{c}^{i}),\left(\pi,c\right)^{(i)})-J_{i}(t,\hat{X}^{i}_{t},\overline{X}^{(i)}_{t},(\pi^{i,t,\epsilon},c^{i,t,\epsilon}),\left(\pi,c\right)^{(i)})\\&=\mathbb{E}_{t}\biggl[\int_{t}^{T}\lambda(s-t)\left[\exp\left\{-\frac{1}{\delta_{i}}(1-\frac{\theta_{i}}{n})c^{i,t,\epsilon}_{s}+\frac{\theta_{i}}{\delta_{i}}\overline{c}^{(i)}_{s}\right\}-\exp\left\{-\frac{1}{\delta_{i}}(1-\frac{\theta_{i}}{n})\hat{c}^{i}_{s}+\frac{\theta_{i}}{\delta_{i}}\overline{c}^{(i)}_{s}\right\}\right]ds\\&+\lambda(T-t)\left[\exp\left\{-\frac{1}{\delta_{i}}(1-\frac{\theta_{i}}{n})X^{i,t,\epsilon}_{T}+\frac{\theta_{i}}{\delta_{i}}\overline{X}^{(i)}_{T}\right\}-\exp\left\{-\frac{1}{\delta_{i}}(1-\frac{\theta_{i}}{n})\hat{X}^{i}_{T}+\frac{\theta_{i}}{\delta_{i}}\overline{X}^{(i)}_{T}\right\}\right]\biggr].
	\end{align*}
Using the concavity and the terminal condition in  BSDE \cref{BSDE1}, we obtain
\begin{equation}
	\begin{split}
		\Delta(J_{i})&\geq\mathbb{E}_{t}\biggl[\int_{t}^{T}\left\langle\lambda(s-t)\frac{1}{\delta_{i}}\left(1-\frac{\theta_{i}}{n}\right)\exp\left\{-\frac{1}{\delta_{i}}\left(1-\frac{\theta_{i}}{n}\right)\hat{c}^{i}_{s}+\frac{\theta_{i}}{\delta_{i}}\overline{c}^{(i)}_{s}\right\},\hat{c}^{i}_{s}-c^{i,t,\epsilon}_{s}\right\rangle ds\\&+\lambda(T-t)\left\langle p(T),\hat{X}^{i}_{T}-X^{i,t,\epsilon}_{T}\right\rangle\biggr].\label{ineq}
	\end{split}
\end{equation}
Applying Ito's formula, we arrive at 
\begin{equation}
	\begin{split}\label{ito}
		\mathbb{E}_{t}\left[\left\langle {p(T)},\hat{X}^{i}_{T}-X^{i,t,\epsilon}_{T}\right\rangle\right]&=\mathbb{E}_{t}\biggl[\int_{t}^{T}\left\langle \mu_{i}{p(s)}+\nu_{i}{q^{i}(s)}+\sigma_{i}{q^{B}(s)},\hat{\pi}^{i}_{s}-\pi^{i,t,\epsilon}_{s}\right\rangle\\&+\left\langle-{p(s)},\hat{c}^{i}_{s}-c^{i,t,\epsilon}_{s}\right\rangle ds\biggr].
	\end{split}
\end{equation}
 Taking \cref{ito} in \cref{ineq} yields 
\begin{align*}
&J_{i}(t,\hat{X}^{i}_{t},\overline{X}^{(i)}_{t},(\pi^{i,t,\epsilon},c^{i,t,\epsilon}),\left(\pi,c\right)^{(i)})-J_{i}(t,\hat{X}^{i}_{t},\overline{X}^{(i)}_{t},(\hat{\pi}^{i},\hat{c}^{i}),\left(\pi,c\right)^{(i)})\\\leq&\mathbb{E}_{t}\biggl[\int_{t}^{T}\left\langle{-\lambda(T-t)p(s)}+\lambda(s-t)\frac{1}{\delta_{i}}\left(1-\frac{\theta_{i}}{n}\right)\exp\left\{-\frac{1}{\delta_{i}}\left(1-\frac{\theta_{i}}{n}\right)\hat{c}^{i}_{s}+\frac{\theta_{i}}{\delta_{i}}\overline{c}^{(i)}_{s}\right\},c^{i,t,\epsilon}_{s}-\hat{c}^{i}_{s}\right\rangle\\+&\lambda(T-t)\langle\mu_{i}{p(s)}+\nu_{i}{q^{i}(s)}+\sigma_{i}{q^{B}(s)},\pi^{i,t,\epsilon}_{s}-\hat{\pi}^{i}_{s}\rangle ds\biggr]\\=&\mathbb{E}_{t}\biggl[\int_{t}^{t+\epsilon}\left\langle{-\lambda(T-t)p(s)}+\lambda(s-t)\frac{1}{\delta_{i}}\left(1-\frac{\theta_{i}}{n}\right)\exp\left\{-\frac{1}{\delta_{i}}\left(1-\frac{\theta_{i}}{n}\right)\hat{c}^{i}_{s}+\frac{\theta_{i}}{\delta_{i}}\overline{c}^{(i)}_{s}\right\},v_{2}\right\rangle\\&+\lambda(T-t)\langle\mu_{i}{p(s)}+\nu_{i}{q^{i}(s)}+\sigma_{i}{q^{B}(s)},v_{1}\rangle ds\biggr].
\end{align*}
According to the conditions \cref{5.2} and \cref{5.3}, we get
\begin{align*}
	J_{i}&(t,\hat{X}^{i}_{t},{\overline{X}^{(i)}_{t}},(\pi^{i,\epsilon},c^{i,\epsilon}),\left(\pi,c\right)^{(i)})-J_{i}(t,\hat{X}^{i}_{t},{\overline{X}^{(i)}_{t}},(\hat{\pi}^{i},\hat{c}^{i}),\left(\pi,c\right)^{(i)})\\&\leq\mathbb{E}_{t}\biggl[\int_{t}^{t+\epsilon}\langle(\lambda(s-t)\lambda(T-s)-\lambda(T-t)){p(s)},v_{2}\rangle\\&+\lambda(T-t)\langle\mu_{i}{p(s)}+\nu_{i}{q^{i}(s)}+\sigma_{i}{q^{B}(s)},v_{1}\rangle ds\biggr]\\&=\mathbb{E}_{t}\biggl[\int_{t}^{t+\epsilon}\langle(\lambda(s-t)\lambda(T-s)-\lambda(T-t)){p(s)},v_{2}\rangle ds\biggr]\\&\leq\max_{t\leq s\leq t+\epsilon}|\lambda(s-t)\lambda(T-s)-\lambda(T-t)|\mathbb{E}_{t}\left[\sup_{s\in[t,T]}{|p(s)|}\right]|v_{2}|\epsilon ,\quad\, a.s.
\end{align*}

Now dividing both sides of the last inequality by $\epsilon$ and taking the limit, as $\lambda$ is continuous,  we conclude that $(\hat{\pi}^{i},\hat{c}^{i})$ is an open-loop consistent control for agent $i$ with respect to the {initial condition $(t_{0},x_{0})$} in response to $(\pi^{k},c^{k})_{k\neq i}$.
\end{proof}
In the view of  \cref{suff}, we consider the following FBSDEs:
	\begin{equation}
		\begin{cases}
			d\hat{X}^{i}_{t}=(\hat{\pi}^{i}_{t}\mu_{i}-\hat{c}^{i}_{t})dt+\hat{\pi}^{i}_{t}\nu_{i}dW^{i}_{t}+\hat{\pi}^{i}_{t}\sigma_{i}dB_{t},\,{t\in[t_{0},T]},\\
			d\overline{X}^{(i)}_{t}=(\overline{\mu\pi}^{(i)}_{t}-\overline{c}^{(i)}_{t})dt+\overline{\sigma\pi}^{(i)}_{t}dB_{t}+\frac{1}{n}\sum_{k\ne i}\nu_{k}\pi^{k}_{t}dW^{k}_{t}, \,{t\in[t_{0},T]},\\
			dp(t)=q^{B}(t)dB_{t}+\sum_{k=1}^{n}q^{i}(t)dW^{k}_{t},\, {t\in[t_{0},T]},\\ \hat{X}^{i}_{{t_{0}}}=x^{i}_{0},\,\overline{X}^{(i)}_{{t_{0}}}=\overline{x}^{(i)}_{0},\\
			p(T)=\frac{1}{\delta_{i}}\left(1-\frac{\theta_{i}}{n}\right)\exp\left\{-\frac{1}{\delta_{i}}\left(1-\frac{\theta_{i}}{n}\right)\hat{X}_{T}^{i}+\frac{\theta_{i}}{\delta_{i}}\overline{X}^{(i)}_{T}\right\},
		\end{cases}\label{FB}
	\end{equation}
with conditions 
\begin{align}
	&\mu_{i}{p(t)}+\nu_{i}{q^{i}(t)}+\sigma_{i}{q^{B}(t)}=0,\,\ {t\in[t_{0},T]},\label{suff1}\\
	&{-\lambda(T-t)p(t)}+\frac{1}{\delta_{i}}\left(1-\frac{\theta_{i}}{n}\right)\exp\left\{-\frac{1}{\delta_{i}}\left(1-\frac{\theta_{i}}{n}\right)\hat{c}^{i}_{t}+\frac{\theta_{i}}{\delta_{i}}\overline{c}^{(i)}_{t}\right\}=0,\,\ {t\in[t_{0},T]}.\label{suff2}
\end{align}
It is hard to solve the above system due to the non-Markovianity of the inputs $(\pi^{k},c^{k})_{k\neq i}$.  Hence, we assume that  the inputs $(\pi^{k},c^{k})_{k\neq i}$ are of the following form:
\begin{equation*}
	\left(\pi^{k}_{t},c^{k}_{t}\right)=\left(\Pi^{k}(t),C^{k}(t,X_{t})\right),\,\ k\neq i,\, {t\in[t_{0},T]},
\end{equation*}
where  $\left(\Pi^{k},C^{k}\right)_{k\neq i}\in\mathcal{S}^{n-1}$ and  $X=(X^{1},\cdots,X^{n})$ is the wealth process associated with $\left(\hat{\pi}^{i},\hat{c}^{i}\right)$ and $(\pi^{k},c^{k})_{k\neq i}$ satisfying
	\begin{equation}
	\begin{cases}
		dX^{k}_{t}=\left(\Pi^{k}(t)\mu_{k}-C^{k}(t,X_{t})\right)dt+\Pi^{k}(t)\nu_{k}dW^{k}_{t}+\Pi^{k}(t)\sigma_{k}dB_{t},\, k\neq i,\, {t\in[t_{0},T]}, \\
		dX^{i}_{t}=\left(\hat{\pi}_{t}^{i}\mu_{i}-\hat{c}^{i}_{t}\right)dt+\hat{\pi}_{t}^{i}\nu_{i}dW_{t}^{i}+\hat{\pi}_{t}^{i}\sigma _{i}dB_{t},\, {t\in[t_{0},T]},\\ X^{j}_{{t_{0}}}=x^{j}_{0},\, j=1,\dots,n.
	\end{cases}
\end{equation}
Now, the constrained FBSDEs bocomes
\begin{equation}\label{A.12}
	\begin{cases}
		d\overline{X}^{(i)}_{t}=\left(\overline{\mu\Pi}^{(i)}(t)-\overline{C}^{(i)}(t,X_{t})\right)dt+\overline{\sigma\Pi}^{(i)}(t)dB_{t}+\frac{1}{n}\sum_{k\ne i}\nu_{k}\Pi^{k}(t)dW^{k}_{t}, \\
	d\hat{X}^{i}_{t}=(\hat{\pi}^{i}_{t}\mu_{i}-\hat{c}^{i}_{t})dt+\hat{\pi}^{i}_{t}\nu_{i}dW^{i}_{t}+\hat{\pi}^{i}_{t}\sigma_{i}dB_{t},\\
	{dp(t)=q^{B}(t)dB_{t}+\sum_{k=1}^{n}q^{i}(t)dW^{k}_{t}},\\ \hat{X}^{i}_{{t_{0}}}=x^{i}_{0},\,\overline{X}^{(i)}_{{t_{0}}}=\overline{x}^{(i)}_{0},\\
	{p(T)=\frac{1}{\delta_{i}}\left(1-\frac{\theta_{i}}{n}\right)\exp\left\{-\frac{1}{\delta_{i}}\left(1-\frac{\theta_{i}}{n}\right)\hat{X}_{T}^{i}+\frac{\theta_{i}}{\delta_{i}}\overline{X}^{(i)}_{T}\right\}},
	\end{cases}
\end{equation}
with conditions 
\begin{align}
	&\mu_{i}{p(t)}+\nu_{i}{q^{i}(t)}+\sigma_{i}{q^{B}(t)}=0,\\
	&{-\lambda(T-t)p(t)}+\frac{1}{\delta_{i}}\left(1-\frac{\theta_{i}}{n}\right)\exp\left\{-\frac{1}{\delta_{i}}\left(1-\frac{\theta_{i}}{n}\right)\hat{c}^{i}_{t}+\frac{\theta_{i}}{\delta_{i}}\overline{C}^{(i)}(t,X_{t})\right\}=0,
\end{align}
where
\begin{equation*}
	 \overline{\mu\Pi}^{(i)}(t):=\frac{1}{n}\sum_{k\neq i}\mu_{k}\Pi^{k}(t),\,\ \overline{\sigma\Pi}^{(i)}(t):=\frac{1}{n}\sum_{k\neq i}\sigma_{k}\Pi^{k}(t),\,\ \overline{C}^{(i)}(t,x)=\frac{1}{n}\sum_{k\neq i}C^{k}(t,x).
\end{equation*}

Based on the terminal condition of the BSDE \cref{BSDE1}, we consider the following ansatz:
\begin{equation}
	p(t)=V^{i}(t,\hat{X}^{i}_{t},\overline{X}^{(i)}_{t}),\,\, {t\in[t_{0},T]},\label{An1}
\end{equation}
for some deterministic function $V^{i}\in C^{1,2,2}([0,T]\times\mathbb{R}^{2})$ satisfying the following terminal condition: \[V^{i}(T,x^{i},\overline{x}^{(i)})=\frac{1}{\delta_{i}}\left(1-\frac{\theta_{i}}{n}\right)\exp\left\{-\frac{1}{\delta_{i}}\left(1-\frac{\theta_{i}}{n}\right)x^{i}+\frac{\theta_{i}}{\delta_{i}}\overline{x}^{(i)}\right\}.\]
Applying Ito's formula to \cref{An1} yields 
\begin{align*}
	dp(&t)=\left\{V^{i}_{t}(t,\hat{X}^{i}_{t},\overline{X}^{(i)}_{t})+V_{x}^{i}(t,\hat{X}^{i}_{t},\overline{X}^{(i)}_{t})(\hat{\pi}^{i}_{t}\mu_{i}-\hat{c}^{i}_{t})+V^{i}_{y}(t,\hat{X}^{i}_{t},\overline{X}^{(i)}_{t})\left(\overline{\mu\Pi}^{(i)}(t)\right.\right.\\&\left. \left.-\overline{C}^{(i)}(t,X_{t})\right)+\frac{1}{2}V^{i}_{xx}(t,\hat{X}^{i}_{t},\overline{X}^{(i)}_{t})(\hat{\pi}^{i}_{t})^{2}(\nu_{i}^{2}+\sigma_{i}^{2})+\frac{1}{2}V^{i}_{yy}(t,\hat{X}^{i}_{t},\overline{X}^{(i)}_{t})\left[\left(\overline{\sigma\Pi}^{(i)}(t)\right)^{2}\right.\right.\\&\left.\left.+\frac{1}{n^{2}}\sum_{k\neq i}\left(\nu_{k}\Pi^{k}(t)\right)^{2}\right]+V^{i}_{xy}(t,\hat{X}^{i}_{t},\overline{X}^{(i)}_{t})\sigma_{i}\hat{\pi}^{i}_{t}\overline{\sigma\Pi}^{(i)}(t)\right\}dt+\left(V^{i}_{x}(t,\hat{X}^{i}_{t},\overline{X}^{(i)}_{t})\sigma_{i}\hat{\pi}^{i}_{t}\right.\\&\left.+V^{i}_{y}(t,\hat{X}^{i}_{t},\overline{X}^{(i)}_{t})\overline{\sigma\Pi}^{(i)}(t)\right)dB_{t}+\frac{1}{n}\sum_{k\neq i}V^{i}_{y}(t,\hat{X}^{i}_{t},\overline{X}^{(i)}_{t})\nu_{k}\Pi^{k}(t)dW^{k}_{t}\\&+V^{i}_{x}(t,\hat{X}^{i}_{t},\overline{X}^{(i)}_{t})\nu_{i}\hat{\pi}^{i}_{t}dW^{i}_{t}.
\end{align*}

Comparing the above equation with the third equation in \cref{A.12}, we deduce that
\begin{equation}\label{pde}
	\begin{cases}
		V^{i}_{t}(t,x^{i},\overline{x}^{(i)})+\hat{\Pi}^{i}(t,x)\left(V^{i}_{x}(t,x^{i},\overline{x}^{(i)})\mu_{i}+V^{i}_{xy}(t,x^{i},\overline{x}^{(i)})\sigma_{i}\overline{\sigma\Pi}^{(i)}(t)\right)\\+\frac{1}{2}\hat{\Pi}^{i}(t,x)^{2}V^{i}_{xx}(t,x^{i},\overline{x}^{(i)})\left(\nu_{i}^{2}+\sigma_{i}^{2}\right)-\hat{C}^{i}(t,x)V^{i}_{x}(t,x^{i},\overline{x}^{(i)})\\+V^{i}_{y}(t,x^{i},\overline{x}^{(i)})\left(\overline{\mu\Pi}^{(i)}(t)-\overline{C}^{(i)}(t,x)\right)\\+\frac{1}{2}V^{i}_{yy}(t,x^{i},\overline{x}^{(i)})\left[\left(\overline{\sigma\Pi}^{(i)}(t)\right)^{2}+\frac{1}{n^{2}}\sum_{k\neq i}\left(\nu_{k}\Pi^{k}(t)\right)^{2}\right]=0, \\
      {q^{B}(t)=\left(V^{i}_{x}(t,\hat{X}^{i}_{t},\overline{X}^{(i)}_{t})\sigma_{i}\hat{\pi}^{i}_{t}+V^{i}_{y}(t,\hat{X}^{i}_{t},\overline{X}^{(i)}_{t})\overline{\sigma\Pi}^{(i)}(t)\right)},\\ {q^{k}(t)={\mathbf{1}}_{\{k=i\}}V^{i}_{x}(t,\hat{X}^{i}_{t},\overline{X}^{(i)}_{t})\nu_{i}\hat{\pi}^{i}_{t}+{\mathbf{1}}_{\{k\ne i\}}\frac{1}{n}V^{i}_{y}(t,\hat{X}^{i}_{t},\overline{X}^{(i)}_{t})\nu_{k}\Pi^{k}(t)}.
\end{cases}
\end{equation}

Putting the second and third equations of \cref{pde} into \cref{suff1} and \cref{suff2}, we have
\begin{align*}
	&\mu_{i}V^{i}(t,\hat{X}^{i}_{t},\overline{X}^{(i)}_{t})+\nu_{i}^{2}V^{i}_{x}(t,\hat{X}^{i}_{t},\overline{X}^{(i)}_{t})\hat{\pi}^{i}_{t}+\sigma_{i}^{2}V^{i}_{x}(t,\hat{X}^{i}_{t},\overline{X}^{(i)}_{t})\hat{\pi}^{i}_{t}+\sigma_{i}\overline{\sigma\Pi}^{(i)}(t)V^{i}_{y}(t,\hat{X}^{i}_{t},\overline{X}^{(i)}_{t})=0, \\
	&\lambda(T-t)V^{i}(t,\hat{X}^{i}_{t},\overline{X}^{(i)}_{t})=\frac{1}{\delta_{i}}\left(1-\frac{\theta_{i}}{n}\right)\exp\left\{-\frac{1}{\delta_{i}}\left(1-\frac{\theta_{i}}{n}\right)\hat{c}^{i}_{t}+\frac{\theta_{i}}{\delta_{i}}\overline{C}^{(i)}(t,X_{t})\right\}, 
\end{align*}
which implies that $\left(\hat{\pi}^{i},\hat{c}^{i}\right)$ is in feedback form:
\begin{align}
	&\hat{\pi}^{i}_{t}=\hat{\Pi}^{i}(t,X_{t})=-\frac{\mu_{i}V^{i}(t,\hat{X}^{i}_{t},\overline{X}^{(i)}_{t})+\sigma_{i}\overline{\sigma\Pi}^{(i)}(t)V^{i}_{y}(t,\hat{X}^{i}_{t},\overline{X}^{(i)}_{t})}{(\nu_{i}^{2}+\sigma_{i}^{2})V^{i}_{x}(t,\hat{X}^{i}_{t},\overline{X}^{(i)}_{t})},\label{pi}\\
	&\hat{c}^{i}_{t}=\hat{C}^{i}(t,X_{t})=-\frac{\delta_{i}}{1-\frac{\theta_{i}}{n}}\ln\left[\frac{\delta_{i}}{1-\frac{\theta_{i}}{n}}\lambda(T-t)V^{i}(t,\hat{X}^{i}_{t},\overline{X}^{(i)}_{t})\right]+\frac{\theta_{i}}{1-\frac{\theta_{i}}{n}}\overline{C}^{(i)}(t,X_{t}).\label{c}
\end{align}

Then, substituting the expressions \cref{pi} and \cref{c} into the first equation of  (\ref{pde}), we obtain the following verification theorem.
\begin{theorem}\label{verif}
        Under the assumption of \cref{suff}, if there exists a classical solution $V^{i}\in C^{1,2,2}((0,T)\times\mathbb{R}^{2})\cap C([0,T]\times\mathbb{R}^{2})$ of the following PDE:
\begin{equation}\label{A.19}
\begin{cases}
	V^{i}_{t}(t,x^{i},\overline{x}^{(i)})+\hat{\Pi}^{i}(t,x)\left(V^{i}_{x}(t,x^{i},\overline{x}^{(i)})\mu_{i}+V^{i}_{xy}(t,x^{i},\overline{x}^{(i)})\sigma_{i}\overline{\sigma\Pi}^{(i)}(t)\right)\\+\frac{1}{2}\hat{\Pi}^{i}(t,x)^{2}V^{i}_{xx}(t,x^{i},\overline{x}^{(i)})\left(\nu_{i}^{2}+\sigma_{i}^{2}\right)-\hat{C}^{i}(t,x)V^{i}_{x}(t,x^{i},\overline{x}^{(i)})\\+V^{i}_{y}(t,x^{i},\overline{x}^{(i)})\left(\overline{\mu\Pi}^{(i)}(t)-\overline{C}^{(i)}(t,x)\right)\\+\frac{1}{2}V^{i}_{yy}(t,x^{i},\overline{x}^{(i)})\left[\left(\overline{\sigma\Pi}^{(i)}(t)\right)^{2}+\frac{1}{n^{2}}\sum_{k\neq i}\left(\nu_{k}\Pi^{k}(t)\right)^{2}\right]=0, (t,x)\in[0,T]\times\mathbb{R}^{n},\\
	V^{i}(T,x^{i},\overline{x}^{(i)})=\frac{1}{\delta_{i}}\left(1-\frac{\theta_{i}}{n}\right)\exp\left\{-\frac{1}{\delta_{i}}\left(1-\frac{\theta_{i}}{n}\right)x^{i}+\frac{\theta_{i}}{\delta_{i}}\overline{x}^{(i)}\right\},\,x\in\mathbb{R}^{n},\\
	\hat{\Pi}^{i}(t,x)=-\frac{\mu_{i}V^{i}(t,x^{i},\overline{x}^{(i)})+\sigma_{i}\overline{\sigma\Pi}^{(i)}(t)V^{i}_{y}(t,x^{i},\overline{x}^{(i)})}{(\nu_{i}^{2}+\sigma_{i}^{2})V^{i}_{x}(t,x^{i},\overline{x}^{(i)})},\,(t,x)\in[0,T]\times\mathbb{R}^{n},\\
	\hat{C}^{i}(t,x)=-\frac{\delta_{i}}{1-\frac{\theta_{i}}{n}}\ln\left[\frac{\delta_{i}}{1-\frac{\theta_{i}}{n}}\lambda(T-t)V^{i}(t,x^{i},\overline{x}^{(i)})\right]\\+\frac{\theta_{i}}{1-\frac{\theta_{i}}{n}}\overline{C}^{(i)}(t,x),\,(t,x)\in[0,T]\times\mathbb{R}^{n},
\end{cases}
\end{equation}
and if the ($n$-dimensional) SDE
	\begin{equation}\label{ndsde}
	\begin{cases}
		dX^{k}_{t}=\left(\Pi^{k}(t)\mu_{k}-C^{k}(t,X_{t})\right)dt+\Pi^{k}(t)\nu_{k}dW^{k}_{t}+\Pi^{k}(t)\sigma_{k}dB_{t},\, k\neq i,\, \\  {t\in[t_{0},T]}, \\
		dX^{i}_{t}=\left(\hat{\Pi}^{i}(t,X_{t})\mu_{i}-\hat{C}^{i}(t,X_{t})\right)dt+\hat{\Pi}^{i}(t,X_{t})\nu_{i}dW_{t}^{i}+\hat{\Pi}^{i}(t,X_{t})\sigma_{i}dB_{t},\,\\ {t\in[t_{0},T]},\quad X^{j}_{{t_{0}}}=x^{j}_{0},\, j=1,\dots,n,
	\end{cases}
\end{equation}
has a solution $X=(X^{1},\dots,X^{n})$ such that the controls {$(\pi^{j}_{t},c^{j}_{t}):=(\Pi^{j}(t),C^{j}(t,X_{t}))$}, $j\neq i$, and $(\hat{\pi}^{i}_{t}, \hat{c}^{i}_{t}):=(\hat{\Pi}^{i}(t,X_{t}),\hat{C}^{i}(t,X_{t}))$ are admissible over $[{t_{0}},T]$,  then  $(\hat{\pi}^{i},\hat{c}^{i})$ is an open-loop consistent control for agent $i$  with respect to the {initial condition $(t_{0},x_{0})$} in response to $(\pi^{k},c^{k})_{k\neq i}$.
\end{theorem}

We stress that in order to prove that $\left(\hat{\pi}^{i}, \hat{c}^{i}\right)$ is an open-loop consistent control, one should  verify the following conditions:
\begin{itemize}
\item[(i)] The SDE \cref{ndsde} has a solution;

\item[(ii)] $\left((\hat{\pi}^{i},\hat{c}^{i}),(\pi,c)^{(i)}\right)\in{\mathcal{A}^{n}_{t_{0}}}$ and $U_{i}(\hat{X}_{T}^{i},\overline{X}^{(i)}_{T})\in L^{2}_{\mathcal{F}_{T}}(\Omega;\mathbb{R})$.
\end{itemize}

The following lemma shows that if $\left((\hat{\Pi}^{i},\hat{C}^{i}),(\Pi,C)^{(i)}\right)\in\mathcal{S}^{n}$, then conditions (i) and (ii) are satisfied.

\begin{lemma}\label{dfadmissiable}
	If $\left((\hat{\Pi}^{i},\hat{C}^{i}),(\Pi,C)^{(i)}\right)\in\mathcal{S}^{n}$, then the SDE \cref{ndsde} has a unique solution. Moreover, the corresponding outcome $\left((\hat{\pi}^{i},\hat{c}^{i}),(\pi,c)^{(i)}\right)$ is in {$\mathcal{A}_{t_{0}}^{n}$} and satisfies the conditions of \cref{suff} and \cref{Assumption2.3}, {i.e., $U_{i}(\hat{X}_{T}^{i},\overline{X}^{(i)}_{T})\in L^{2}_{\mathcal{F}_{T}}(\Omega;\mathbb{R})$ and $\left((\hat{\pi}^{i},\hat{c}^{i}),(\pi,c)^{(i)}\right)\in I^{n}_{t_{0},x_{0}}$.}
\end{lemma}
\begin{proof}
	Under the DF strategy $\left((\hat{\Pi}^{i},\hat{C}^{i}),(\Pi,C)^{(i)}\right)\in\mathcal{S}^{n}$, the  SDE becomes a linear SDE as follows:
	\begin{equation}\label{newsde}
		\begin{cases}
			dX_{t}=\left(A(t)X_{t}+b(t)\right)dt+D(t)dW_{t},\quad {t\in[t_{0},T]},\\                
			X_{{t_{0}}}=x_{0},
		\end{cases}
	\end{equation}
where $X_{t}:=(X^{1}_{t},\dots,X^{n}_{t})^{\top}$, $W_{t}:=\left(W^{1}_{t},\dots,W^{n}_{t}, B_{t}\right)^{\top}$, $x_{0}:=\left(x^{1}_{0}\dots,x^{n}_{0}\right)^{\top}$ and $A(t):[0,T]\rightarrow\mathbb{R}^{n\times n}$, $b(t):[0,T]\rightarrow\mathbb{R}^{n}$, $D(t):[0,T]\rightarrow\mathbb{R}^{n\times (n+1)}$ are  deterministic continuous functions.

Then, we  get that \cref{newsde} has a unique solution $X\in\mathbb{S}^{2}_{\mathbb{F}}({t_{0}},T;\mathbb{R}^{n})$. Moreover, we have the following estimate:
\begin{equation}
	\mathbb{E}\left[\sup_{t\in[{t_{0}},T]}|X_{t}|^{2}\right]\leq C
\end{equation}
for some constant $C$. Then, it is obvious to see the outcome $\left((\hat{\pi}^{i},\hat{c}^{i}),(\pi,c)^{(i)}\right)$ are admissible over $[{t_{0}},T]$.

In fact, we can solve the  SDE \cref{newsde} explicitly:
\begin{equation}
	X_{t}=\Phi(t)x_{0}+\Phi(t)\int_{{t_{0}}}^{t}\Phi^{-1}(s)b(s)ds+\Phi(t)\int_{{t_{0}}}^{t}\Phi^{-1}(s)D(s)dW_{s},
\end{equation}
where $\Phi(t)$ is the unique solution of the following ODE:
\begin{equation*}
	\begin{cases}
		d\Phi(t)=A(t)\Phi(t)dt,\\
		\Phi({t_{0}})=I_{n\times n},
	\end{cases}
\end{equation*}
and $\Phi^{-1}(t)$ exists, satisfying
\begin{equation*}
	\begin{cases}
		d\Phi^{-1}(t)=-\Phi^{-1}(t)A(t)dt,\\
		\Phi^{-1}({t_{0}})=I_{n\times n}.
	\end{cases}
\end{equation*}
 We now claim a stronger result: for $\beta^{\top}(t):=\left(\beta_{1}(t),\dots,\beta_{n}(t)\right)$  an $n$-dimensional continuous function, the process $\left\{\exp\left(\beta^{\top}(t)X_{t}\right), {t_{0}}\leq t\leq T\right\}$ is uniformly integrable. 
It suffices to prove that $	\sup_{t\in[{t_{0}},T]}\mathbb{E}\left[\exp\left(p\beta^{\top}(t)X_{t}\right)\right]<\infty$ for some $p>1$.
 
Note that $\exp\left(p\beta^{\top}(t)X_{t}\right)$ has the following form
\begin{equation*}
	\exp\left(p\beta^{\top}(t)\Phi(t)x_{0}+p\beta^{\top}(t)\Phi(t)\int_{{t_{0}}}^{t}\Phi^{-1}(s)b(s)ds+p\beta^{\top}(t)\Phi(t)\int_{{t_{0}}}^{t}\Phi^{-1}(s)D(s)dW_{s}\right),
\end{equation*}
and it is enough to focus on $\exp\left(p\beta^{\top}(t)\Phi(t)\int_{{t_{0}}}^{t}\Phi^{-1}(s)D(s)dW_{s}\right)$.

Since $Y_{t}:=\int_{{t_{0}}}^{t}\Phi^{-1}(s)D(s)dW_{s}\sim N(0,\Sigma(t))$, where $\Sigma(t):[{t_{0}},T]\rightarrow\mathbb{R}^{n\times n}$ is a continuous function, we can obtain
\begin{equation*}
	\begin{split}
		\mathbb{E}\left[\exp\left(p\beta^{\top}(t)\Phi(t)\int_{{t_{0}}}^{t}\Phi^{-1}(s)D(s)dW_{s}\right)\right]&=\mathbb{E}\left[\exp\left(p\beta^{\top}(t)\Phi(t)Y_{t}\right)\right]
\\&=\exp\left(\frac{p^{2}\beta^{\top}(t)\Phi(t)\Sigma(t)\Phi^{\top}(t)\beta(t)}{2}\right).	\end{split}
\end{equation*}
 It follows that 
\begin{equation*}
	\begin{split}
		\sup_{t\in[{t_{0}},T]}\mathbb{E}\left[\exp\left(p\beta^{\top}(t)X_{t}\right)\right]=\max_{t\in[{t_{0}},T]}&\exp\left(p\beta^{\top}(t)\Phi(t)x_{0}+p\beta^{\top}(t)\Phi(t)\int_{{t_{0}}}^{t}\Phi^{-1}(s)b(s)ds\right.\\&\left.+\frac{p^{2}\beta^{\top}(t)\Phi(t)\Sigma(t)\Phi^{\top}(t)\beta(t)}{2}\right)<\infty,
	\end{split}
\end{equation*}
which proves our claim. Therefore, for every $i\in\left\{1,\cdots,n\right\}$,
\begin{equation*}
	\begin{split}
&\mathbb{E}\left[\int_{{t_{0}}}^{T}|U_{i}(c^{i}_{t},\overline{c}^{(i)}_{t})|^{2}dt\right]=\int_{{t_{0}}}^{T}\mathbb{E}\left[\exp\left(-\frac{2}{\delta_{i}}\left(1-\frac{\theta_{i}}{n}\right)\left(\sum_{k=1}^{n}p^{i,k}(t)X^{k}_{t}+q^{i}(t)\right)\right.\right.\\&\left.\left.+\sum_{j\ne i}\frac{2\theta_{i}}{n\delta_{i}}\left(\sum_{k=1}^{n}p^{j,k}(t)X^{k}_{t}+q^{j}(t)\right)\right)\right]dt<\infty,\\
&\mathbb{E}\left[|U_{i}(X^{i}_{T},\overline{X}^{(i)}_{T})|^{2}\right]=\mathbb{E}\left[\exp\left\{-\frac{2}{\delta_{i}}\left(1-\frac{\theta_{i}}{n}\right)X_{T}^{i}+2\frac{\theta_{i}}{\delta_{i}}\overline{X}^{(i)}_{T}\right\}\right]<\infty.
\end{split}
\end{equation*}

Hence, $\left((\hat{\pi}^{i},\hat{c}^{i}),(\pi,c)^{(i)}\right)$ satisfies the conditions of \cref{suff} and \cref{Assumption2.3}, which completes the proof.
\end{proof}

\subsection{The mean field game}

Similar to  the $n$-agent games, we can deduce a characterization of open-loop consistent control in the mean field game framework.
Here, we only give the corresponding result and omit the proof.

Let $\xi_{0}=(\delta_{0},\theta_{0},\mu_{0},\nu_{0},\sigma_{0})$ represent a deterministic sample from its random type distribution and take an initial pair $(t_{0},x_{0})\in[0,T)\times\mathbb{R}$.  We denote by $(\hat{\pi}^{\xi_{0}},\hat{c}^{\xi_{0}})\in\mathcal{A}^{MF}_{t_{0}}$ a candidate control.  Now, we introduce the following BSDE defined on the interval $[t_{0},T]$:
\begin{equation}\label{bsde5.1}
	\begin{cases}
		dp(t)=q^{B}(t)dB_{t}+q^{W}(t)dW_{t}, \\
		p(T)=\frac{1}{\delta_{0}}\exp\left(-\frac{1}{\delta_{0}}\hat{X}^{\xi_{0}}_{T}+\frac{\theta_{0}}{\delta_{0}}\overline{X}_{T}\right),
	\end{cases}
\end{equation}
where $\hat{X}^{\xi_{0}}$ is given by
\begin{equation}\label{Xxi0}
	\begin{cases}
		d\hat{X}^{\xi_{0}}_{t}=\left(\hat{\pi}^{\xi_{0}}_{t}\mu_{0}-\hat{c}^{\xi_{0}}_{t}\right)dt+\hat{\pi}^{\xi_{0}}_{t}\nu_{0}dW_{t}+\hat{\pi}^{\xi_{0}}_{t}\sigma_{0}dB_{t},\,t\in[t_{0},T],\\
		\hat{X}^{\xi_{0}}_{t_{0}}=x_{0},
	\end{cases}
\end{equation}
and $\overline{X}$ is defined by \cref{MFEmeanwealth}.
\begin{theorem}\label{suffmean}
	Assume that  $\exp\left(-\frac{1}{\delta_{0}}\hat{X}^{\xi_{0}}_{T}+\frac{\theta_{0}}{\delta_{0}}\overline{X}_{T}\right)\in L^{2}_{\mathcal{F}^{MF}_{T}}(\Omega;\mathbb{R})$.  Then, $(\hat{\pi}^{\xi_{0}},\hat{c}^{\xi_{0}})$ is an open-loop consistent control for $\xi_{0}$-type agent with respect to the initial condition $(t_{0},x_{0})$, if the following conditions hold
	\begin{align*}
		&\mu_{0} p(t)+\nu_{0} q^{W}(t)+\sigma_{0} q^{B}(t)=0,\quad t\in[t_{0},T],\\
		&-\lambda(T-t)p(t)+\frac{1}{\delta_{0}}\exp\left(-\frac{1}{\delta_{0}}\hat{c}^{\xi_{0}}_{t}+\frac{\theta_{0}}{\delta_{0}}\overline{C}(t,\overline{X}_{t})\right)=0,\quad t\in[t_{0},T],
	\end{align*}
where $\overline{C}(t,\overline{X}_{t}):=\left(p_{1}(t)+\mathbb{E}\left[p_{2}(t,\xi)\right]\right)\overline{X}_{t}+\mathbb{E}\left[q(t,\xi)\right]$.
\end{theorem}  

Similarly, by considering the following FBSDEs:
\begin{equation*}
	\begin{cases}
		d\hat{X}^{\xi_{0}}_{t}=\left(\hat{\Pi}^{\xi_{0}}(t,\hat{X}^{\xi_{0}}_{t},\overline{X}_{t})\mu_{0}-\hat{C}^{\xi_{0}}(t,\hat{X}^{\xi_{0}}_{t},\overline{X}_{t})\right)dt+\hat{\Pi}^{\xi_{0}}(t,\hat{X}^{\xi_{0}}_{t},\overline{X}_{t})\nu_{0} dW_{t}\\+\hat{\Pi}^{\xi_{0}}(t,\hat{X}^{\xi_{0}}_{t},\overline{X}_{t})\sigma_{0} dB_{t},\,t\in[t_{0},T],\\
	    d\overline{X}_{t}=\left[-\left(p_{1}(t)+\mathbb{E}\left[p_{2}(t,\xi)\right]\right)\overline{X}_{t}+\mathbb{E}\left[\Pi^{\xi}(t)\mu\right]-\mathbb{E}\left[q(t,\xi)\right]\right]dt\\+E\left[\Pi^{\xi}(t)\sigma\right]dB_{t},\,t\in[t_{0},T],\\
	    dp(t)=q^{B}(t)dB_{t}+q^{W}(t)dW_{t},\,\  t\in[t_{0},T],\\ 
	    \hat{X}^{\xi_{0}}_{t_{0}}=x_{0},\,\overline{X}_{t_{0}}=x_{0},\,
		p(T)=\frac{1}{\delta_{0}}\exp\left(-\frac{1}{\delta_{0}}\hat{X}^{\xi_{0}}_{T}+\frac{\theta_{0}}{\delta_{0}}\overline{X}_{T}\right),
	\end{cases}\label{FBM}
\end{equation*}
with conditions 
\begin{align}
	&\mu_{0}p(t)+\nu_{0}q^{W}(t)+\sigma_{0}q^{B}(t)=0,\quad t\in[t_{0},T],\\
	&-\lambda(T-t)p(t)+\frac{1}{\delta_{0}}\exp\left(-\frac{1}{\delta_{0}}\hat{c}^{\xi_{0}}_{t}+\frac{\theta_{0}}{\delta_{0}}\overline{C}(t,\overline{X}_{t})\right)=0,\quad t\in[t_{0},T],\label{suffmean2}
\end{align}
we have the following verification theorem in the mean field game framework.

\begin{theorem}\label{verifmean}
Under the assumption of \cref{suffmean}, if there exists a classical solution $V^{\xi_{0}}\in C^{1,2,2}((0,T)\times\mathbb{R}^{2},\mathbb{R})\cap C([0,T]\times\mathbb{R}^{2},\mathbb{R})$ of the following PDE:
\begin{equation*}
	\begin{cases}
		V^{\xi_{0}}_{t}(t,x,\overline{x})+\hat{\Pi}^{\xi_{0}}(t,x,\overline{x})\left(V^{\xi_{0}}_{x}(t,x,\overline{x})\mu_{0}+V^{\xi_{0}}_{x\overline{x}}(t,x,\overline{x})\sigma_{0}\mathbb{E}\left[\Pi^{\xi}(t)\sigma\right]\right)\\+\frac{1}{2}\hat{\Pi}^{\xi_{0}}(t,x,\overline{x})^{2}V^{\xi_{0}}_{xx}(t,x,\overline{x})\left(\nu_{0}^{2}+\sigma_{0}^{2}\right)-\hat{C}^{\xi_{0}}(t,x,\overline{x})V^{\xi_{0}}_{x}(t,x,\overline{x})\\+V^{\xi_{0}}_{\overline{x}}(t,x,\overline{x})\left(\mathbb{E}\left[\Pi^{\xi}(t)\mu\right]-\overline{C}(t,\overline{x})\right)\\+\frac{1}{2}V^{\xi_{0}}_{\overline{x}\overline{x}}(t,x,\overline{x})\left(\mathbb{E}\left[\Pi^{\xi}(t)\sigma\right]\right)^{2}=0,\quad (t,x,\overline{x})\in[0,T]\times\mathbb{R}^{2},\\
		V^{\xi_{0}}(T,x,\overline{x})=\frac{1}{\delta_{0}}\exp\left(-\frac{1}{\delta_{0}}x+\frac{\theta_{0}}{\delta_{0}}\overline{x}\right),\quad(x,\overline{x})\in\mathbb{R}^{2},\\
		\hat{\Pi}^{\xi_{0}}(t,x,\overline{x})=-\frac{\mu_{0} V^{\xi_{0}}(t,x,\overline{x})+\sigma_{0}\mathbb{E}\left[\Pi^{\xi}(t)\sigma\right]V^{\xi_{0}}_{\overline{x}}(t,x,\overline{x})}{\left(\nu_{0}^{2}+\sigma_{0}^{2}\right)V^{\xi_{0}}_{x}(t,x,\overline{x})},\quad (t,x,\overline{x})\in[0,T]\times\mathbb{R}^{2},\\
		\hat{C}^{\xi_{0}}(t,x,\overline{x})=-\delta_{0}\ln\left[\delta_{0}\lambda(T-t)V^{\xi_{0}}(t,x,\overline{x})\right]+\theta_{0}\overline{C}(t,\overline{x}),\quad (t,x,\overline{x})\in[0,T]\times\mathbb{R}^{2},
	\end{cases}
\end{equation*}
and if the SDE 
\begin{equation}\label{8.28}
	\begin{cases}
			d\hat{X}^{\xi_{0}}_{t}=\left(\hat{\Pi}^{\xi_{0}}(t,\hat{X}^{\xi_{0}}_{t},\overline{X}_{t})\mu_{0}-\hat{C}^{\xi_{0}}(t,\hat{X}^{\xi_{0}}_{t},\overline{X}_{t})\right)dt+\hat{\Pi}^{\xi_{0}}(t,\hat{X}^{\xi_{0}}_{t},\overline{X}_{t})\nu_{0} dW_{t}\\+\hat{\Pi}^{\xi_{0}}(t,\hat{X}^{\xi_{0}}_{t},\overline{X}_{t})\sigma_{0} dB_{t},\quad t\in[t_{0},T],\\
		d\overline{X}_{t}=\left[-\left(p_{1}(t)+\mathbb{E}\left[p_{2}(t,\xi)\right]\right)\overline{X}_{t}+\mathbb{E}\left[\Pi^{\xi}(t)\mu\right]-\mathbb{E}\left[q(t,\xi)\right]\right]dt\\+E\left[\Pi^{\xi}(t)\sigma\right]dB_{t},\quad t\in[t_{0},T],\\
		\hat{X}^{\xi_{0}}_{t_{0}}=x_{0},\,\overline{X}_{t_{0}}=x_{0}.
	\end{cases}
\end{equation}
has a solution $\left(\hat{X}^{\xi_{0}},\overline{X}\right)$ such that $(\hat{\pi}^{\xi_{0}}_{t},\hat{c}^{\xi_{0}}_{t}):=\left(\hat{\Pi}^{\xi_{0}}(t,\hat{X}^{\xi_{0}}_{t},\overline{X}_{t}),\hat{C}^{\xi_{0}}(t,\hat{X}^{\xi_{0}}_{t},\overline{X}_{t})\right)$ are admissible over $[t_{0},T]$, then  $(\hat{\pi}^{\xi_{0}},\hat{c}^{\xi_{0}})$ is an open-loop consistent control for $\xi_{0}$-type agent with respect to the initial condition $(t_{0},x_{0})$. 
\end{theorem}

\section{Caculation on $\hat{h}^{i}(t)$}\label{h}

Recalling that
\begin{equation}
		\hat{h}^{i}(t):=\frac{1}{T+1-t}\int_{t}^{T}(T+1-s)\left[G^{i}((\Pi^{*,k}(s))_{k\ne i},s)\right]ds,
\end{equation}
where \begin{align*}
	G^{i}((\pi^{k})_{k\ne i},t)&=-\frac{1}{2}\frac{\left[\mu_{i}+\sigma_{i}\overline{\sigma\pi}^{(i)}g^{i}(t)\right]^{2}}{\nu_{i}^{2}+\sigma_{i}^{2}}-\frac{1}{T+1-t}\ln\left[\lambda(T-t)\right]+g^{i}(t)\overline{\mu\pi}^{(i)}\\&+\frac{1}{2}\left[g^{i}(t)\right]^{2}\big[(\overline{\sigma\pi}^{(i)})^{2}+\frac{1}{n^{2}}\sum_{k\neq i}(\nu_{k}\pi^{k})^{2}\big]
\end{align*}
and $$g^{i}(t)=\frac{\theta_{i}}{\delta_{i}}\frac{1}{\left(T+1-t\right)}.$$

 Let $A^{i}_{n}:=\overline{\sigma\Pi^{*}}^{(i)}(t)g^{i}(t)$, $B^{i}_{n}:=\overline{\mu\Pi^{*}}^{(i)}(t)g^{i}(t)$ and $C^{i}_{n}:=\frac{1}{n^{2}}\sum_{k\neq i}\left(\nu_{k}\Pi^{*,k}(t)g^{i}(t)\right)^{2}$. One can check that $A^{i}_{n},B^{i}_{n}$ and $C^{i}_{n}$ are constants:
\begin{equation}
	\begin{split}
		A^{i}_{n}&=\frac{1}{n}\frac{\theta_{i}}{\delta_{i}}\sum_{k\ne i}\left[\delta_{k}\frac{\sigma_{k}\mu_{k}}{\sigma_{k}^{2}+\left(1-\frac{\theta_{k}}{n}\right)\nu_{k}^{2}}+\theta_{k}\frac{\sigma_{k}^{2}}{\sigma_{k}^{2}+\left(1-\frac{\theta_{k}}{n}\right)\nu_{k}^{2}}\frac{\phi_{n}}{1-\psi_{n}}\right],\\
		B^{i}_{n}&=\frac{1}{n}\frac{\theta_{i}}{\delta_{i}}\sum_{k\ne i}\left[\delta_{k}\frac{\mu_{k}^{2}}{\sigma_{k}^{2}+\left(1-\frac{\theta_{k}}{n}\right)\nu_{k}^{2}}+\theta_{k}\frac{\mu_{k}\sigma_{k}}{\sigma_{k}^{2}+\left(1-\frac{\theta_{k}}{n}\right)\nu_{k}^{2}}\frac{\phi_{n}}{1-\psi_{n}}\right],\\
		C^{i}_{n}&=(\frac{1}{n}\frac{\theta_{i}}{\delta_{i}})^{2}\sum_{k\ne i}\left[\delta_{k}\frac{\nu_{k}\mu_{k}}{\sigma_{k}^{2}+\left(1-\frac{\theta_{k}}{n}\right)\nu_{k}^{2}}+\theta_{k}\frac{\nu_{k}\sigma_{k}}{\sigma_{k}^{2}+\left(1-\frac{\theta_{k}}{n}\right)\nu_{k}^{2}}\frac{\phi_{n}}{1-\psi_{n}}\right]^{2}.
	\end{split}	
\end{equation}
Then \[G^{i}((\Pi^{*,k}(t))_{k\ne i},t)=-\frac{1}{2}\frac{(\mu_{i}+\sigma_{i}A^{i}_{n})^{2}}{\nu_{i}^{2}+\sigma_{i}^{2}}-\frac{1}{T+1-t}\ln\left[\lambda(T-t)\right]+B^{i}_{n}+\frac{1}{2}[(A^{i}_{n})^{2}+C^{i}_{n}].\]
Hence, 
\begin{equation}
	\hat{h}^{i}(t)=\frac{D^{i}_{n}}{2}\left[\frac{1}{T+1-t}-(T+1-t)\right]-\frac{1}{T+1-t}\int_{t}^{T}\ln\left[\lambda(T-s)\right]ds,
\end{equation}
where
\begin{equation}
D_{n}^{i}:=\frac{1}{2}\frac{(\mu_{i}+\sigma_{i}A^{i}_{n})^{2}}{\nu_{i}^{2}+\sigma_{i}^{2}}-\frac{1}{2}[(A^{i}_{n})^{2}+C^{i}_{n}]-B^{i}_{n}.
\end{equation} 

If we assume that empirical measure $m_{n}$ has a weak limit $m$, by passing to limit, we then have
\begin{equation*}
	\begin{split}
	A^{i}:&=\lim_{n\rightarrow\infty} A^{i}_{n}=\frac{\theta_{i}}{\delta_{i}}\mathbb{E}\left[\delta\frac{\sigma\mu}{\sigma^{2}+\nu^{2}}+\theta\frac{\sigma^{2}}{\sigma^{2}+\nu^{2}}\frac{\phi}{1-\psi}\right]=\frac{\theta_{i}}{\delta_{i}}\frac{\phi}{1-\psi}=g^{i}(t)\mathbb{E}[\sigma\Pi^{*,\xi}(t)],\\
	B^{i}:&=\lim_{n\rightarrow\infty}B^{i}_{n}=\frac{\theta_{i}}{\delta_{i}}\mathbb{E}\left[\delta\frac{\mu^{2}}{\sigma^{2}+\nu^{2}}+\theta\frac{\sigma\mu}{\sigma^{2}+\nu^{2}}\frac{\phi}{1-\psi}\right]=g^{i}(t)\mathbb{E}[\mu\Pi^{*,\xi}(t)],\\
	C^{i}:&=\lim_{n\rightarrow\infty}C^{i}_{n}=0,\\
	D^{i}:&=\lim_{n\rightarrow\infty}D^{i}_{n}=\frac{1}{2}\frac{\left(\mu_{i}+\sigma_{i}g^{i}(t)\mathbb{E}[\sigma\Pi^{*,\xi}(t)]\right)^{2}}{\nu_{i}^{2}+\sigma_{i}^{2}}-\frac{1}{2}\left(g^{i}(t)\mathbb{E}[\sigma\Pi^{*,\xi}(t)]\right)^{2}-g^{i}(t)\mathbb{E}[\mu\Pi^{*,\xi}(t)].
	\end{split}
\end{equation*}

Then the limit of $\hat{h}^{i}(t)$ is given by \begin{equation*}
	H^{i}(t):=\lim_{n\rightarrow\infty}\hat{h}^{i}(t)=\frac{D^{i}}{2}\left[\frac{1}{T+1-t}-(T+1-t)\right]-\frac{1}{T+1-t}\int_{t}^{T}\ln\left[\lambda(T-s)\right]ds.
\end{equation*}
 
 Note that
\begin{equation*}
\hat{h}^{\xi}(t)=\frac{1}{T+1-t}\int_{t}^{T}(T+1-s)G^{\xi}(\Pi^{*,\xi}(s),s)ds,
\end{equation*}  
where
\begin{equation*}
	\begin{split}
	G^{\xi}(\pi,t)&=-\frac{1}{T+1-t}\ln\left[\lambda(T-t)\right]-\frac{1}{2}\frac{\left(\mu+\sigma g^{\xi}(t)\mathbb{E}[\sigma\pi]\right)^{2}}{\nu^{2}+\sigma^{2}}\\&+g^{\xi}(t)\mathbb{E}[\mu\pi]+\frac{1}{2}\left(g^{\xi}(t)\right)^{2}(\mathbb{E}[\sigma\pi])^{2},
\end{split}
\end{equation*} 
then \begin{equation}
	\hat{h}^{\xi}(t)=\frac{D}{2}\left[\frac{1}{T+1-t}-(T+1-t)\right]-\frac{1}{T+1-t}\int_{t}^{T}\ln\left[\lambda(T-s)\right]ds=H^{\xi}(t).
\end{equation}
\section*{Acknowledgments}
 The authors thank the anonymous referees for their valuable comments and suggestions, which have significantly improved the quality of the paper. The authors also thank the members of the group of Actuarial Science and Mathematical Finance at  the Department of Mathematical Sciences, Tsinghua University for their feedbacks and useful conversations.

\bibliographystyle{siamplain}
\bibliography{references}

\end{document}